\documentclass[journal,twocolumn]{IEEEtran}
\usepackage{cite}
\usepackage{amsmath,amssymb,amsfonts,amsthm}
\usepackage{graphicx}
\usepackage{multirow}
\usepackage{subfigure}
\usepackage{textcomp}
\usepackage{enumerate,paralist}
\usepackage{mathrsfs} 
\usepackage{setspace}
\usepackage{nomencl}
\usepackage{threeparttable}
\usepackage{bm}
\usepackage{comment}
\usepackage{dblfloatfix}
\usepackage[ruled,linesnumbered]{algorithm2e}
\makenomenclature

\theoremstyle{definition}
\newtheorem{theorem}{\bf Theorem}

\newtheorem{remark}{\bf Remark}
\newtheorem{prop}{\bf Proposition}
\newtheorem{assumption}{Assumption}
\newtheorem{corollary}{Corollary}

\graphicspath{{Graphics/}}	
\usepackage[colorlinks, linkcolor=black, anchorcolor=black, citecolor=black]{hyperref}

\begin{document}
	\title{Multi-kernel Correntropy Regression: Robustness, Optimality, and Application on Magnetometer Calibration}
	\author{Shilei Li, Yihan Chen, Yunjiang Lou, Dawei Shi, Lijing Li, Ling Shi
	\thanks{Manuscript received March 29, 2023. This work was supported in part by the NSFC-Shenzhen Robotics Basic Research Center Program under Grant U1913208 and in part by the Shenzhen Science and Technology Program under Grant JCYJ20220818102415033.}
	\thanks{Shilei Li and Ling Shi are with the Department of Electronic and Computer Engineering, The Hong Kong University of Science and Technology, Hong Kong, China (e-mail: slidk@connect.ust.hk, eesling@ust.hk).}
	\thanks{Yihan Chen and Yunjiang Lou are with the State Key Laboratory of Robotics and System, School of Mechanical Engineering and Automation, Harbin Institute of Technology Shenzhen, Shenzhen 518055, China (e-mail:  200310114@stu.hit.edu.cn, louyj@hit.edu.cn).}
	\thanks{Dawei Shi is with the School of Automation, Beijing Institute of Technology, China (e-mail: daweishi@bit.edu.cn).}
	\thanks{Lijing Li is with the School of Information and Control Engineering, China University of Mining and Technology, China (e-mail: lilijing\_29@163.com).}}
	\maketitle
\begin{abstract}
	This paper investigates the robustness and optimality of the multi-kernel correntropy (MKC) on linear regression. We first derive an upper error bound for a scalar regression problem in the presence of arbitrarily large outliers and state that the kernel bandwidth plays an important role in minimizing the lowest upper error bound. Then, we find that the proposed MKC is related to a specific heavy-tail distribution, where its head shape is consistent with the Gaussian distribution but its tail shape is heavy-tailed, and the extent of heavy-tail is controlled by the kernel bandwidth. It becomes a Gaussian distribution when the bandwidth is infinite, which allows one to tackle both Gaussian and non-Gaussian problems without explicitly investigating the noise distributions. To explore the optimal underlying distribution parameters, an expectation-maximization-like (EM) algorithm is developed to estimate the parameter vectors and the distribution parameters in an alternating manner. The results show that our algorithm can achieve equivalent performance compared with the traditional linear regression under Gaussian noise, and it significantly outperforms the conventional method under heavy-tailed noise. Both numerical simulations and experiments on a magnetometer calibration application verify the effectiveness of the proposed method.
\end{abstract}
\def\abstractname{Note to Practitioners}
\begin{abstract}
	The goal of this paper is to enhance the accuracy of conventional linear regression in handling outliers while maintaining its optimality under Gaussian assumptions. Our algorithm is formulated under the maximum likelihood estimation (MLE) framework, assuming the regression residuals follow a type of heavy-tailed noise distribution with an extreme case of Gaussian. The degree of the heavy tail is explored alternatingly using an Expectation-Maximization (EM) algorithm which converges very quickly. The robustness and optimality of the proposed approach are investigated and compared with the traditional approaches. Both theoretical analysis and experiments on magnetometer calibration demonstrate the superiority of the proposed method over the conventional methods. In the future, we will extend the proposed method to more general cases (such as nonlinear regression and classification) and derive new algorithms to accommodate more complex applications (such as with equality or inequality constraints or with prior knowledge of parameter vectors).
\end{abstract}
\section{Introduction}
   Regression is the procedure of uncovering a mathematical relationship of interest through a set of inputs, outputs, and a known mapping function. Conventional solutions for this problem are some second-order statistics-based algorithms, e.g., weighted least square (WLS) regression and its variants~\cite{a1,a2}. However, they perform poorly under non-Gaussian noise~\cite{a3}, especially in the presence of outliers or disturbances since the underlying assumption behind the WLS is Gaussian noise distribution. \textcolor{black}{One practical non-Gaussian noise example is the ellipsoid calibration of the magnetometer, where the sensor is vulnerable to ferromagnetic materials and can be easily distorted by surrounding ferromagnetic materials. In such a scenario, our aim is to recover the ellipsoid parameter vector even if the measured data is polluted~\cite{ab4}.} The non-Gaussian noise is also very common in many other practical engineering problems. It can be caused by intermittent sensor failures, communication disruptions, external disturbances, and multipath effects of signals (e.g., fault diagnosis in ~\cite{a4}, identification of switched linear systems in ~\cite{a5}, \textcolor{black}{multipath effects in \cite{ab5}}, etc.), and hence should be taken into consideration when designing algorithms.

   Many robust techniques have been developed to accommodate the heavy-tailed non-Gaussian noise, which roughly can be divided into two categories: robust statistics and correntropy. Some typical methods of the first category include the least trimmed squares~\cite{a6}, the least median of squares\cite{a7}, the least absolute derivation (LAD)~\cite{a8}, the fractional lower order moments~\cite{a9}, and the M-estimators~\cite{a10}. Recently, correntropy, given its root in Renyi's entropy~\cite{a11,a12} under the framework of information-theoretic learning, has emerged. It is a local similarity measure of two random variables where the kernel bandwidth acts as a zoom lens controlling the ``observation window'' in which similarity is assessed~\cite{a11}. The correntropy has the ability to capture a higher order of error moments~\cite{a13} and has
   been successfully applied to regression~\cite{a14,a15}, kernel adaptive filtering~\cite{a16}, state estimation~\cite{a17}, smoothing~\cite{a18}, adaptive filtering~\cite{a19}, and machine learning~\cite{a20}. It is worth mentioning that the correntropy is a non-convex objective function of residuals. Existing solutions to this problem include the gradient descent~\cite{a11,b20}, fixed-point iteration~\cite{a13,c20,a17,a18}, half-quadratic methods~\cite{d20}, and evolutionary algorithms~\cite{e20}.   
   
   The correntropy-based methods generally can enhance the robustness of regression with respect to heavy-tailed noise, but this ability is closely related to the kernel bandwidth of the kernel function which should be optimized based on the characteristic of the data set. There are two fundamental questions on correntropy that need to be explored: how robust it is and how to select the kernel bandwidth? To the best of the authors' knowledge, only \cite{a14,a21} discussed the robustness of the correntropy. In ~\cite{a14}, an explicit error bound was derived under a errors-in-variables (EIVs) model with scalar variables. In \cite{a21}, a general robustness analysis for linear regression was presented. Unfortunately, the bound presented in \cite{a21} was not computable. Some works discussed the selection criteria of the kernel bandwidths, which include~\cite{a21,a17,a22,a23}. In ~\cite{a21,a17}, an adaptive kernel bandwidth which is proportional to the amplitude of the error was employed. This strategy is usually deployed for the convenience of practical implementation and the optimality is not guaranteed. In ~\cite{a22}, the kernel bandwidth was updated iteratively by seeking the greatest attenuation along the direction of the gradient ascent. In \cite{a23}, a probability density matching (pdf) strategy was utilized to explore the kernel sizes. However, to the best of the author's knowledge, these parameter selection strategies majorly are developed by intuitions or empirical experience and cannot guarantee optimality under the framework of maximum likelihood estimation.
   
   In this work, we handle the aforementioned questions under the framework of \emph{multi-kernel correntropy} (MKC), which is an extension of the original correntropy proposed in our previous works~\cite{a24,a25,a26}. There are two major differences between the MKC and conventional correntropy. The first is that we use different kernel bandwidths for different random pair of variables (we denotes them as different channels in the following section for convenience) which greatly alleviates the conservatism of conventional correntropy. This modification can be an analogy of using heteroscedastic loss to replace the homoscedastic loss in optimization. The second is that specific weights are associated with the MKC so that the MKC-induced distribution becomes the Gaussian distribution with infinite kernel bandwidth. In this paper, we first provide a fixed-point solution for linear regression under the MKCL (i.e., MKC loss). Then, we derive an upper error bound for the MKCL in a scalar regression problem and prove that the MKCL is much more robust to outliers compared with the WLS. Further, we disclose that the MKCL is associated with a specific type of heavy-tailed distribution where its head shape is determined by the corresponding Gaussian distribution, and its tail shape is controlled by the kernel bandwidth. This finding provides a clear relationship between the correntropy and its induced noise distribution and makes it possible to optimize correntropy parameters under the framework of MLE (note that it is equivalent to minimizing the dissimilarity between the empirical distribution defined by the training set and the model distribution, with the degree of dissimilarity between the two measured by the KL divergence~\cite{b12}). Interestingly, the MKCL-associated distribution is equivalent to the Gaussian distribution when the kernel bandwidth is infinite, indicating that the MKCL-based algorithm is always at least as effective as the WLS-based approaches when the kernel bandwidth is properly selected. To automatically adjust the correntropy parameters, we develop an EM-like algorithm that alternatingly estimates the kernel parameters and the parameter vector to maximize the overall log-likelihood function. 
   
   We conducted both numerical simulations and experiments to verify the effectiveness of the proposed algorithm. Specifically, two numerical simulations were performed to demonstrate the proposed algorithm's robustness and superiority. In addition, we  conducted an experiment of ellipsoid fitting for magnetometer calibration to verify the effectiveness of the proposed method in a practical application. It is worth noting that ellipsoid fitting with outliers is not only important in sensor calibration~\cite{b5,b6}, but also has significant applications in computer vision~\cite{b2}, robotics, geology, and medical imaging. The contributions of this paper are summarized as follows:
   \color{black}
   \begin{enumerate}
	\item We build an explicit relationship between the MKCL and a type of heavy-tailed distribution. The results indicate the MKCL-based method generally outperforms the WLS solution if the correntropy parameters are properly selected since its induced distribution has an additional free parameter to match the noise tail shape compared with the Gaussian distribution. 
	\item To analyze the robustness of the MKCL, we establish an explicit upper error bound for a scalar regression problem. We find that the derived error bound is closely related to the selection of the correntropy parameters.
	\item To jointly optimize both the correntropy parameters and the parameter vector, an MKC expectation maximization (MKC-EM) algorithm is constructed which optimizes the target state and latent state alternatingly. A fixed-point solution is utilized to estimate the parameter vector under current correntropy parameters. Then, the Broyden–Fletcher–Goldfarb–Shanno (BFGS) method is employed to update the correntropy parameters by assuming that the parameter vector is known. The proposed MKC-EM algorithm converges to the steady state after 2-3 iterations and its superiority over the existing method is verified under both simulations and experiments.
   \end{enumerate} 
   \color{black}
   The remainder of this paper is arranged as follows. In Section II, we present some preliminaries. In Section III, we provide the linear regression under the MKCL and give its robustness analysis and correntropy parameters optimization strategy. In Section IV, we present some illustrative examples and experiments. In Section V, we draw a conclusion. 
   
   \emph{Notations}: The transpose of a matrix $A$ is denoted by $A^{T}$. The vector with $l$ dimensions is denoted by $\mathbb{R}^{l}$ and the matrix with $m$ rows and $n$ columns is denoted by $\mathbb{R}^{m \times n}$. The Gaussian distribution with mean $\mu$ and covariance $\Sigma$ is denoted by $\mathcal{N}(\mu,\Sigma)$. The uniform distribution with bounds $a$ and $b$ is denoted by $\mathcal{U}(a,b)$. The $p$ norm of a vector $x$ or matrix $A$ is denoted by $\|x\|_p$ or $\|A\|_p$. The expectation of a random variable $X$ or random vector $\mathcal{X}$ is denoted by $E(X)$ or $E(\mathcal{X})$. The operator $\operatorname{diag}(\{\cdot\})$ generates a (block) diagonal matrix with the enclosed arguments on the main diagonal.
   \section{Preliminaries}
   In this section, we start from the traditional linear regression under the WLS criterion. Then, we provide some preliminaries of the MKC. \textcolor{black}{Finally, we provide an overview of the proposed method.}
   \subsection{Linear Regression}
   Let $y_k \in \mathbb{R}^{m}$, $X_k \in \mathbb{R}^{m \times n}$ be the output and input of some stochastic processes. They are related by 
   \begin{equation}
   y_k = X_k \theta^{o} + v_k
   \label{sys}
   \end{equation}
   where the subscript $k$ denotes the time index, $v_k \in \mathbb{R}^{m}$ is the noise,  and $\theta^{o} \in \mathbb{R}^{n}$ is the unknown parameter vector. Assume that a total of $N$ samples are available. Then, denote
   $\mathbf{Y}=[y_1^{T},y_2^{T},\ldots,y_{N}^{T}]^{T} \in \mathbb{R}^{mN \times 1}$, $\mathbf{X}=[X_1^{T},X_2^{T},\ldots,X_{N}^{T}]^{T} \in \mathbb{R}^{mN \times n}$ and $\mathbf{V}=[v_1^{T},v_2^{T},\ldots,v_{N}^{T}]^{T} \in \mathbb{R}^{mN \times 1}$, one has 
	\begin{equation}
	\mathbf{Y}=\mathbf{X}\theta^{o}+ \mathbf{V}.
	\label{system}
	\end{equation}
	Let $\theta$ be an estimate of the parameter vector. The difference between the output and the projected values has
	\begin{equation}
	\bm{e}=\mathbf{Y}-\mathbf{X}\theta
	\end{equation}
	where $\bm{e} \in \mathbb{R}^{mN \times 1}$ is the stacked error of $e_k=y_k-X_k \theta$. Under the least square (LS) criterion, the parameter vector can be obtained by solving 
	\begin{equation}
	\begin{aligned}
	&\theta = \arg \min J(\theta)\\
	&J(\theta)=\frac{1}{2N}\bm{e}^{T}\bm{e}.
	\end{aligned}
	\label{jmse}
	\end{equation} 
	Setting $\frac{\partial J(\theta)}{\partial \theta}=0$, one obtains
	\begin{equation}
	\theta=(\mathbf{X}^{T}\mathbf{X})^{-1}\mathbf{X}^{T}\mathbf{Y}.
	\end{equation}
	In some cases, the nominal measurement noise distribution is the \emph{a priori} knowledge which follows $v_k \sim \mathcal{N}(0, R)$ where $R=\operatorname{diag}(\{d_1^2,d_2^2,\ldots,d_m^2\})$ is a diagonal matrix (note that this assumption is without loss of generality since a linear system with a general covariance matrix can be transferred as another linear system with a diagonal covariance matrix using matrix diagonalization technique). The corresponding precision matrix $P$ is the inverse of the covariance matrix with $P=R^{-1}= \operatorname{diag}(\{1/d_1^2,1/d_2^2,\ldots,1/d_m^2\}$). To incorporate this information into regression, the WLS criterion is utilized with
	\begin{equation}
	J(\theta)=\frac{1}{2N}\bm{e}^{T}\mathbf{P}\bm{e},
	\label{WLS}
	\end{equation}
	where $\mathbf{P}=\operatorname{diag}(\{P,P,\ldots,P\}) \in \mathbb{R}^{mN \times mN}$. By setting $\frac{\partial J(\theta)}{\partial \theta}=0$, one has
	\begin{equation}
		\theta=(\mathbf{X}^{T}\mathbf{P}\mathbf{X})^{-1}\mathbf{X}^{T}\mathbf{P}\mathbf{Y}.
		\label{solution}
	\end{equation}
	It is obvious that \eqref{WLS} is identical to \eqref{jmse} if we use the weighted error $\tilde{e}_k = P^{1/2}e_k$ to replace the $e_k$ in \eqref{jmse}. Due to the fact $P$ is diagonal by definition, it follows that $\tilde{e}_k(i)=P_{ii}^{1/2}e_k(i)=\frac{e_k(i)}{d_i}$ where $\tilde{e}_k(i)$ and $e_k(i)$ are the $i$-th element of $\tilde{e}_k$ and $e_k$, and $P_{ii}$ is the $i$-th main diagonal entry of matrix $P$. We use this substitution in the following section for the algorithm derivation.
   \subsection{Multi-kernel Correntropy}
   The correntropy is a local similarity measure of two random variables $X,Y \in \mathbb{R}$ with 
   \begin{equation}\nonumber
   C(X,Y)= E[\kappa(X,Y)]=\int \kappa(x,y) d F_{XY}(x,y)
   \end{equation}
   where $\kappa(x,y)$ is a shift-invariant Mercer kernel, $F_{XY}(x,y)$ is the joint distribution, and $x$ and $y$ are realizations of $X$ and $Y$.
   In ~\cite{a24,a25}, we present the MKC for random vectors $\mathcal{X}, \mathcal{Y} \in \mathbb{R}^{m}$: 
   \begin{equation}\nonumber
   \begin{aligned}
   C(\mathcal{X},\mathcal{Y})&= \sum_{i=1}^{m} E[\sigma_i^2 \kappa_i(\mathcal{X}_i,\mathcal{Y}_i)]\\
   E[\sigma_i^2\kappa_i(\mathcal{X}_i,\mathcal{Y}_i)]&=\int \sigma_i^2\kappa_i\Big(x(i),y(i)\Big)d F_{\mathcal{X}_i\mathcal{Y}_i}\Big(x(i),y(i)\Big)
   \end{aligned}
   \end{equation}
   where $\mathcal{X}_i$ and $\mathcal{Y}_i$ are random elements of $\mathcal{X}$ and $\mathcal{Y}$, $\kappa_i\big(x(i),y(i)\big)=G_{\sigma_i}\big(\tilde{e}(i)\big)=\exp\big(-\frac{\tilde{e}(i)^2}{2\sigma_i^2}\big)$ is the Gaussian kernel, $\sigma_i$ is the kernel bandwidth for random pair $(\mathcal{X}_i,\mathcal{Y}_i)$, $\tilde{e}(i)=\frac{x(i)-y(i)}{d_i}$ is the weighted realization error, and $d_i$ is the nominal standard deviation for channel $i$. In a practical application, joint distribution $F_{\mathcal{X}_i\mathcal{Y}_i}\big(x(i),y(i)\big)$ usually is not available and only $N$ samples can be obtained.  In this situation, one can estimate MKC as
   \begin{equation}
   \begin{aligned}
   \hat{C}(\mathcal{X},\mathcal{Y})&=\sum_{i=1}^{m}\sigma_i^2 \hat{C}_{i}(\mathcal{X}_i,\mathcal{Y}_i)\\
   \hat{C}_{i}(\mathcal{X}_i,\mathcal{Y}_i)&= \frac{1}{N}\sum_{k=1}^{N}G_{\sigma_i}\big(\tilde{e}_k(i)\big)
   \end{aligned}
   \end{equation}
   where $\tilde{e}_{k}=P^{1/2} e_k=P^{1/2}\big(x_{k}-y_{k}\big)$ is the weighted error and $\tilde{e}_{k}(i)$ is the $i$-th element of $\tilde{e}_{k}$. 
   Correspondingly, the MKCL $J_{CL}$ has
   \begin{equation}
   \begin{aligned}
   J_{CL} &=\sum_{i=1}^{m}\sigma_i^2 (1-\hat{C}_{i})\\
   &=\sum_{i=1}^{m}\sigma_i^2 \Big(1-\frac{1}{N}\sum_{k=1}^{N}G_{\sigma_i}\big(\tilde{e}_k(i)\big)\Big).
   \label{GL}
   \end{aligned}
   \end{equation}
   As a comparison, the WLS criterion has
   \begin{equation}
   \begin{aligned}
   J_{WLS} & = \frac{1}{2N} \sum_{k=1}^{N}\tilde{e}_k^{T}\tilde{e}_k\\
   & = \frac{1}{2N}\sum_{i=1}^{m}\sum_{k=1}^{N}\tilde{e}^2_k(i).
   \label{WWLS}
   \end{aligned}
   \end{equation}
   \begin{remark}
   	It is worth mentioning that our proposed MKC has different meanings from the concept proposed in \cite{a20}. From the formulation aspect, we associate specific weight $\sigma_i^2$ for the correntropy at channel $i$ and use different kernel bandwidths for different channels, while \cite{a20} uses a mixture of kernel functions to construct a novel kernel function and apply it to all channels. From the purpose aspect, our aim is to induce a type of heavy-tailed distribution that is consistent with the Gaussian distribution in the extreme case and the extent of heavy tail at different channels can be parameterized by kernel bandwidth while the aim of \cite{a20} is to accommodate complex error distributions (i.e., skewed and multi-peak distributions, see Fig. 1 in \cite{a20}).
   \end{remark}
   \begin{theorem}
   	$J_{GL}$ in \eqref{GL} and $J_{WLS}$ in \eqref{WWLS} are identical when $\sigma_i \to \infty$ for $i=1,2,\ldots,l$.
   	\label{theorem1}
   \end{theorem}
   \begin{proof}
   	Taking Taylor series expansion of $G_{\sigma_i}\big(\tilde{e}_k(i)\big)$, one has
   	\begin{equation}\nonumber
   	G_{\sigma_i}\big(\tilde{e}_k(i)\big)=\sum_{n=0}^{\infty}\frac{(-1)^{n}}{2^{n}{\sigma}_{i}^{2n}n!}\tilde{e}_k^{2 n}(i).
   	\end{equation}
   	When setting $\sigma_i \to \infty$, it follows that
   	\begin{equation}
   	\lim\limits_{\sigma_i \to \infty} \sigma_i^2 \Big(1-G_{\sigma_i}\big(\tilde{e}_k(i)\big)\Big) = \tilde{e}_k^2(i)/2.
   	\label{limtaylor}
   	\end{equation}
   	Substituting \eqref{limtaylor} into \eqref{GL}, one obtains 
   	\begin{equation}
   	\lim\limits_{\sigma_i \to \infty} J_{GL}= \frac{1}{N}\sum_{i=1}^{m}\sum_{k=1}^{N}\tilde{e}_k^2(i)/2=J_{WLS}.
   	\end{equation}
   	This completes the proof.
   \end{proof}
   \begin{remark}
   	It is obvious that the MKCL in \eqref{GL} is closely related to the selection of correntropy parameters $\{\sigma_i, d_i\}_{i=1}^{m}$. We will provide a physical explanation for these parameters in Section \ref{mkcpdf} and optimize them under the MLE framework in Section \ref{kpop}.
   \end{remark}
  \color{black}
 \subsection{Overview of the Proposed Method}
 It is well-known that the performance of the WLS algorithm relies heavily on the Gaussian noise assumption and its performance would degenerate significantly when the training set's distribution is notably different from the model distribution. To solve this issue, we propose an MKC-EM algorithm for linear regression which matches the model distribution and the data distribution automatically. The algorithm overview is summarized in Fig. \ref{overview}, which is composed of a maximization (M) step and an expectation (E) step. The purpose of the M step is to update the parameter vector under the current correntropy parameter. A fixed-point solution is designed to solve this problem and the detailed algorithm is summarized in Algorithm \ref{fps}. The aim of the E step is to update the correntropy parameter so that the model distribution matches the practical noise distribution. A BFGS solution is constructed to handle this problem and the details are in Algorithm \ref{kpo}. The overall method is called MKC-EM and is summarized in Algorithm \ref{lrmkc}. We will provide detailed descriptions of these algorithms in the following section.
 \begin{figure}[!h]
 	\centerline{\includegraphics[width=8.5cm]{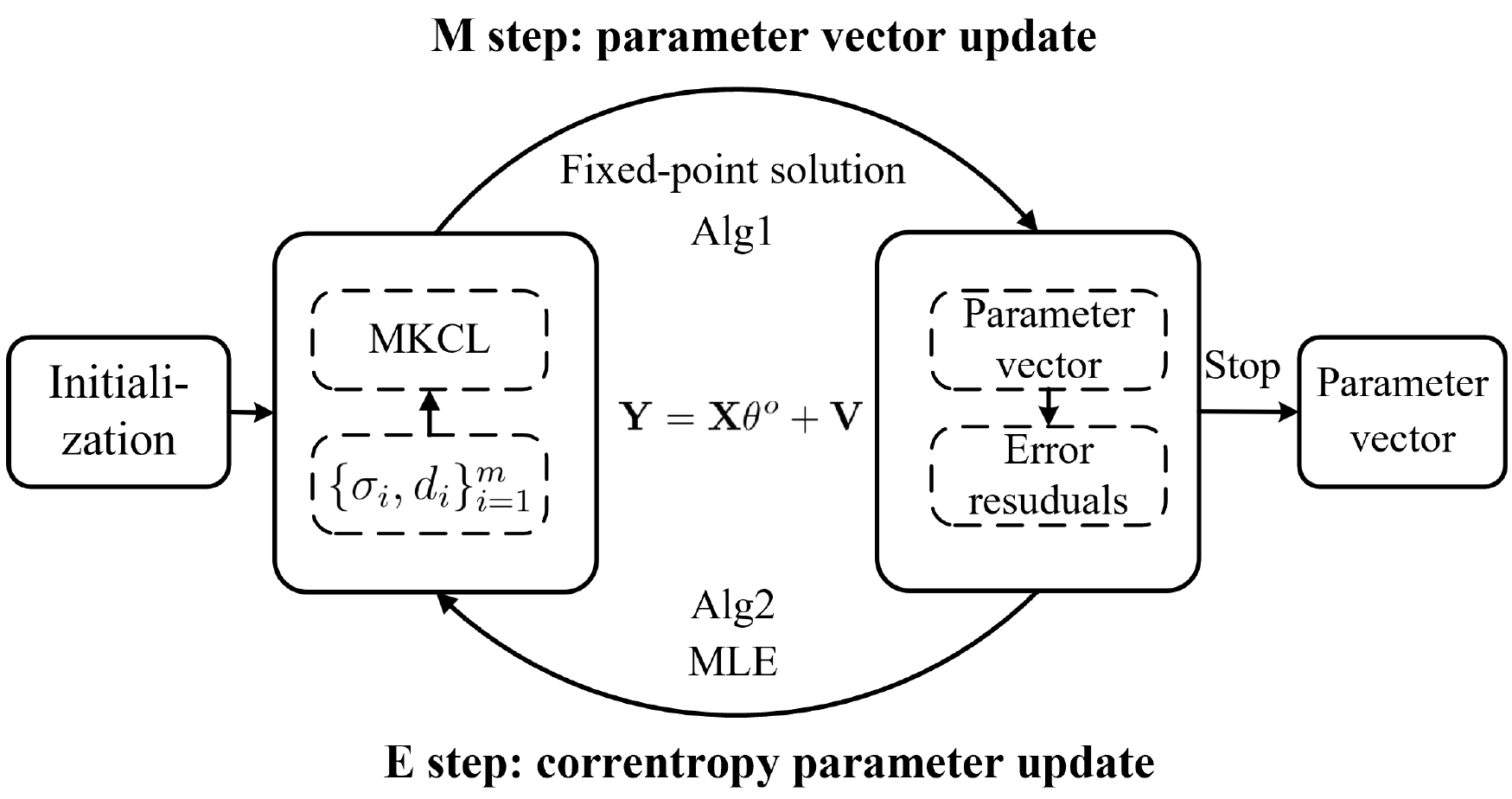}}
 	\caption{\textcolor{black}{Overview of the proposed MKC-EM algorithm for linear regression.}}
 	\label{overview}
 \end{figure}
   \color{black}
   \section{Linear Regression under the MKC} 
   \color{black}
   In this section, we first provide the linear regression solution under the MKCL in Section \ref{mkclreg}. Then, we build an explicit relationship between the MKCL and the noise distribution in Section \ref{mkcpdf}, and provide the robustness analysis of the MKC in Section \ref{robustness}. Furthermore, we give the correntropy parameter optimization strategy and the whole MKC-EM algorithm in Section \ref{kpop}. Finally, we discuss the convergence of the proposed method in Section \ref{convergence}.
  \color{black}
   \subsection{Fixed-point Solution of Linear Regression under the MKCL}
   \label{mkclreg}
    Under the minimum MKCL criterion, the problem \eqref{sys} can be solved by 
   \begin{equation}
   \begin{aligned}
   &\theta= \arg \min J_{CL}(\theta)\\
   &J_{CL}(\theta) = \sum_{i=1}^{m}\sigma_i^2 \Big(1- \frac{1}{N}\sum_{k=1}^{N}G_{\sigma_i}\big(\tilde{e}_{k}(i)\big)\Big)
   \label{cl}
   \end{aligned}
   \end{equation}
   where $\tilde{e}_{k}(i)= \frac{e_k(i)}{d_i}$, $e_k =y_k-X_k \theta$, and $d_i$ and $\sigma_i$ are the nominal standard deviation and the kernel bandwidth for channel $i$. Denote $\mathrm{x}_{k,i} \in \mathbb{R}^{n}$ is the $i$-th row vector of $X_k$, it follows that $e_k(i)= \frac{y_k(i)-\mathrm{x}_{k,i} \theta}{d_i}$.
   Taking partial derivative of \eqref{cl} with respect $\theta$ gives
   \begin{equation}
   \begin{aligned}
   \frac{\partial J_{CL}(\theta)}{\partial \theta}&=-\frac{1}{N}\sum_{k=1}^{N}\sum_{i=1}^{m} \frac{1}{d_i}\mathrm{x}_{k,i}^{T}G_{\sigma_i}\big(\tilde{e}_{k}(i)\big) \tilde{e}_k(i)\\
   &=-\frac{1}{N}\sum_{k=1}^{N}\sum_{i=1}^{m}\frac{1}{d_i^2} \mathrm{x}_{k,i}^{T}G_{\sigma_i}\big(\tilde{e}_{k}(i)\big) \big(y_k(i)-\mathrm{x}_{k,i} \theta\big)\\
   &=-\frac{1}{N}\sum_{k=1}^{N}X_{k}^{T} P W_k(y_k- X_k \theta)
   \end{aligned}
   \end{equation} 
   where $P$ is the precision matrix and $W_k= \operatorname{diag} \big(\{w_1,w_2,\ldots,w_m\}\big)$ with $w_i=G_{\sigma_i}\big(\tilde{e}_{k}(i)\big)$ for $i=1, 2,\ldots, m$. Setting this partial derivative to zero gives
   \begin{equation}
   \theta= \Big[  \sum_{k=1}^{N}(X_k^{T}PW_k X_k) \Big ]^{-1}\cdot \Big[ \sum_{k=1}^{N} X_k^{T}PW_k y_k \Big].
   \label{fixEqu}
   \end{equation}
   Note that both sides of the above equation contain $\theta$ ($W_k$ is a function of $\theta$). Therefore, \eqref{fixEqu} is a fixed-point equation. Then, the fixed-point iteration rule can be utilized (one can refer to \cite{c20} for the convergence of this algorithm), and thus we have 
   \begin{equation}
   \begin{aligned}
   \theta_{t+1}&=f(\theta_{t})\\
   &=\Big[  \sum_{k=1}^{N}(X_k^{T}PW_k X_k) \Big ]^{-1}\cdot \Big[ \sum_{k=1}^{N} X_k^{T}PW_k y_k \Big].
   \label{fixed}
   \end{aligned}
   \end{equation}
   where $t$ is the iteration number starting from 0 and $\theta_0$ is the initial guess of the parameter vector. This algorithm terminates when the parameter vector update is smaller than a predefined threshold $\frac{\|\theta_{t+1}-\theta_{t}\|_2} {\|\theta_t\|_2}\le \psi$. The detailed algorithm is summarized in Algorithm \ref{fps}.
   \begin{algorithm}[t]
   	\caption{Fixed-point solution for linear regression under the MKCL}
   	\label{fps}
   	\KwIn{Input $\{X_k\}_{k=1}^{N}$ and output $\{y_k\}_{k=1}^{N}$}
   	\KwOut{Parameter vector $\theta$}
   	\textbf{Initialization}: initialize $\{d_i\}_{i=1}^{m}$, $\{\sigma_i\}_{i=1}^{m}$, initial guess of the parameter vector $\theta_0$, and a threshold $\psi$\\
   	\While{$\frac{\|\theta_{t+1}-\theta_{t}\|_2} {\|\theta_t\|_2} > \psi$ or $t=0$}
   	{$t \leftarrow t+1$\\
   	$\tiny\theta_{t}=[\sum_{k=1}^{N}(X_k^{T}PW_k X_k) ]^{-1}[\sum_{k=1}^{N} X_k^{T}P W_ky_k]$
   	}
   \end{algorithm}
   \subsection{Pdf Explanation of the MKCL}
   \label{mkcpdf}
     \textcolor{black}{This section builds an explicit relationship between the MKCL and its induced noise distribution, and highlights its connection and distinction with the conventional Gaussian distribution.} It is well known that the WLS criterion in \eqref{WWLS} is optimal with $v_k \sim \mathcal{N}(0,P^{-1})$ in the sense of \emph{maximum likelihood estimation} (MLE). Actually, the MKCL criterion in \eqref{cl} is optimal under the following heavy-tailed distribution.
    \begin{theorem}
   	$J_{CL}$ in \eqref{cl} is optimal in the sense of MLE if $e_{k}(i)$ follows 
   	\begin{equation}
   		p\big(e_k(i)\big)=\frac{c_i}{d_{i}} \exp\Big(-\sigma_i^2\big(1-\exp(-\frac{e_k^2(i)}{2d_{i}^2\sigma_i^2})\big)\Big)
   		\label{pdfe}
   	\end{equation}
   	where $d_i$ is the nominal standard derivation for channel $i$, $\sigma_i$ is the kernel bandwidth, and $c_i$ is a normalization coefficient so that $p\big(e_k(i)\big)$ is a \emph{proper} pdf.
   	\label{theorempdf}
   \end{theorem}
\begin{proof}
	 Under the pdf in \eqref{pdfe}, the likelihood of $\theta$ given  $\{y_k\}_{k=1}^{N}$ and $\{X_k\}_{k=1}^{N}$ follows
	 \begin{equation}\nonumber
	 \begin{aligned}
	 &\mathcal{L}\Big(\theta;\{y_k\}_{k=1}^{N}, \{X_k\}_{k=1}^{N}\Big)=\prod_{k=1}^{N} \prod_{i=1}^{m}p\big({e}_k(i)\big)\\
	 &=\prod_{k=1}^{N} \prod_{i=1}^{m}\frac{c_i}{d_i}\exp\Big(-\sigma_i^2(1-\exp\big(-\frac{{e}_k^2(i)}{2d_i^2\sigma_i^2})\big)\Big).
	 \label{likelihood}
	 \end{aligned}
	 \end{equation}
	 Based on MLE, we have 
	 \begin{equation}\nonumber
	 \theta = \arg \max \mathcal{L}\Big(\theta;\{y_k\}_{k=1}^{N},\{X_{k}\}_{k=1}^{N}\Big).
	 \label{MLE}
	 \end{equation}
	 It is equivalent to minimizing its negative logarithm function with
	 \begin{equation}
	 \begin{aligned}
	 \theta &= \arg \min \sum_{k=1}^{N}\sum_{i=1}^{m} \bigg[-\log(\frac{c_i}{d_i})+\sigma_i^2(1-\exp\big(-\frac{{e}_k^2(i)}{2d_i^2\sigma_i^2})\big)\bigg]\\
	 &\overset{1)}{=}\arg \min \sum_{k=1}^{N}\sum_{i=1}^{m} \sigma_i^2(1-\exp\big(-\frac{{e}_k^2(i)}{2d_i\sigma_i^2})\big)\\
	 &\overset{2)}{=}\arg \min \sum_{i=1}^{m}\sigma_i^2 \Big(1- \frac{1}{N}\sum_{k=1}^{N}G_{\sigma_i}(\tilde{e}_{i}(k))\Big)\\
	 &= J_{CL}.
	 \label{globj}
	 \end{aligned}
	 \end{equation}
	where $1)$ is obtained since $c_i$ and $d_i$ are some constants and $2)$ is obtained by multiplying $1/N$ on the right side of the equation. This completes the proof.
\end{proof}
\begin{corollary}
	If $\sigma_i \to \infty$ in \eqref{pdfe}, $p\big({e}_k(i)\big)$ becomes a Gaussian distribution with $p\big({e}_k(i)\big) = \frac{1}{\sqrt{2\pi}d_i}\exp\Big(-\frac{{e}_{k}^2(i)}{2d_i^2}\Big)$.
\end{corollary}
\begin{proof}
	Based on \eqref{limtaylor} and substituting corresponding results into \eqref{pdfe}, it follows that
	\begin{equation}
	\lim\limits_{\sigma_i \to \infty} p\big({e}_k(i)\big) = \frac{c_i}{d_i} \exp\Big(-\frac{{e}_{k}^2(i)}{2d_i^2}\Big). 
	\end{equation}
	Due to the fact that $\int p\big(e_k(i)\big) d e_k(i) =1 $, one has $c_i = \frac{1}{\sqrt{2\pi}}$. This completes the proof.
\end{proof}
 In one-dimensional case, a comparison of Gaussian distribution and $p(e_k)$ with $d$ and $\sigma$ is shown in Fig. \ref{pdf_ek}. One can see that its head shape is determined by $d$ while its tail shape is controlled by $\sigma$. We also observe that $p(e_k)$ approaches a Gaussian distribution when $\sigma$ is big and this evidence is consistent with Theorem \ref{theorem1}. This property makes the MKCL more versatile than the WLS since it can tackle both Gaussian and heavy-tailed problems by selecting the kernel bandwidth properly.
\begin{figure}[htbp]
	\centering
	\subfigure[]{\includegraphics[width=1.6in]{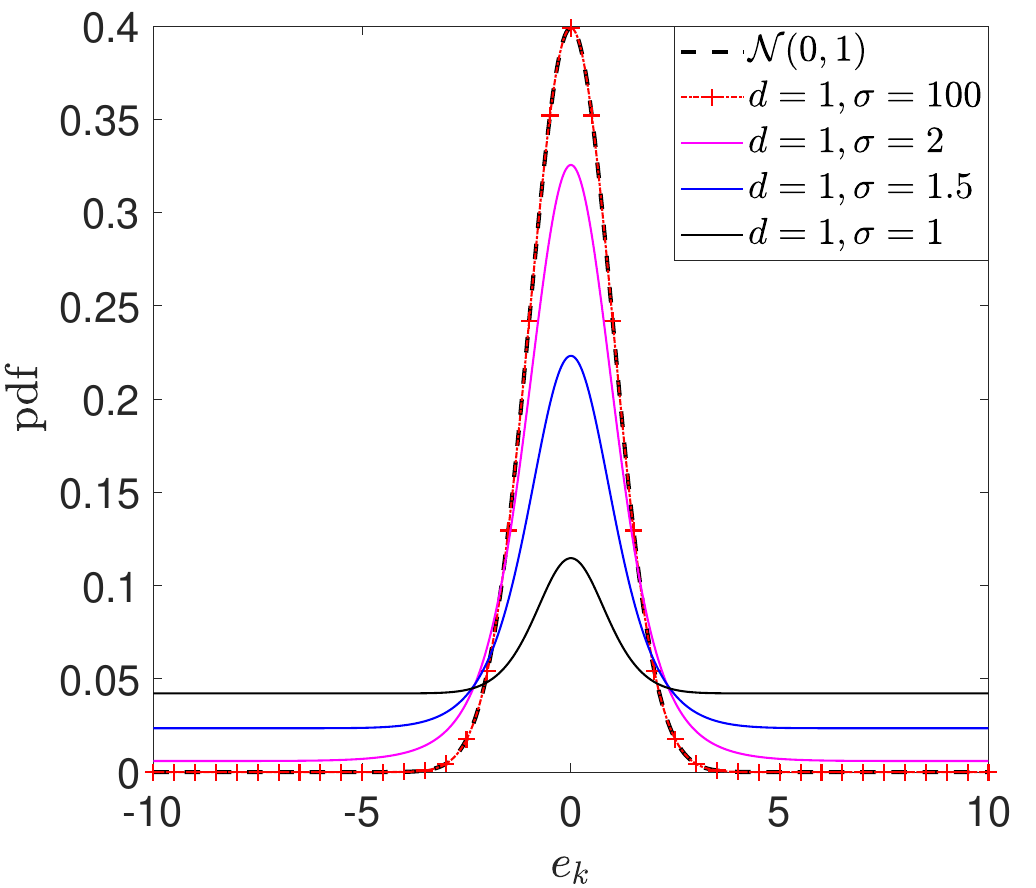}%
		\label{pdf_ek1}}
	\subfigure[]{\includegraphics[width=1.6in]{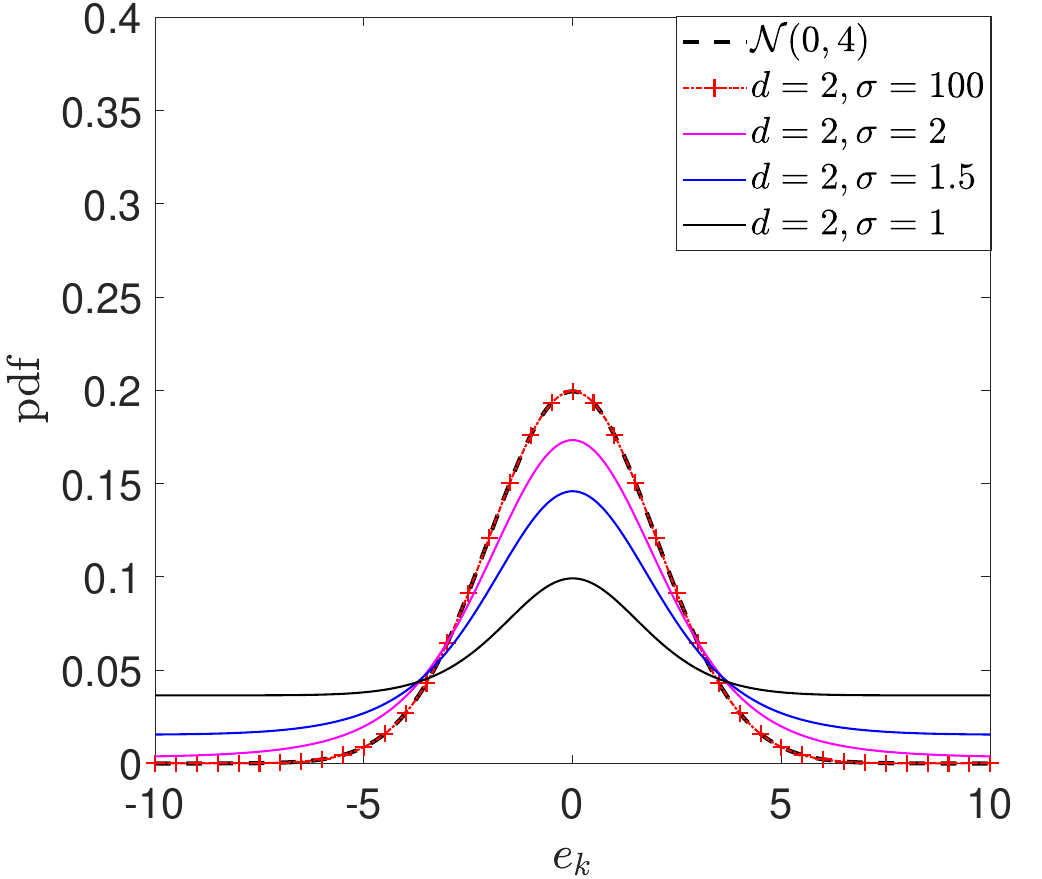}%
		\label{pdf_ek2}}
	\caption{A comparison of Gaussian distribution and $p(e_k)$ in \eqref{pdfe} with different $d$ and $\sigma$. The error $e_k$ is assumed to be bounded within $[-10,10]$ and the coefficient $c$ is obtained by $c=1/\int_{-10}^{10} \frac{1}{d}\exp\Big(-\sigma^2\big(1-\exp(-\frac{e_k^2}{2d^2\sigma^2})\big)\Big)$.}
	\label{pdf_ek}
\end{figure}
\begin{remark}
	As visualized in Fig. \ref{pdf_ek}, the MKCL is a highly suitable loss function when the error distribution is constructed by a mixture of Gaussian distribution and uniform distribution. Additionally, it is an appealing choice when the heavy tail varies across different channels. In a practical implementation, a truncated distribution of $p\big(e_k\big)$ is utilized (i.e., the feasible domain of the error is bounded) so that $c_i$ can be numerically calculated. Note that this distribution approaches a uniform distribution when $\sigma_i \to 0_{+}$ and the corresponding question would become an $\ell_0$ norm optimization problem which is N-P hard.
\end{remark}

The likelihood function is a way to measure how well a statistical model explains the training set. To illustrate the versatility of the MKCL over the WLS, we assume that the practical residuals $e_k$ follow 
\begin{equation}
	e_k \sim (1-p)\mathcal{N}(0,1) + p\mathcal{U}(-20,20), 0 \le p<0.5
	\label{residual}
\end{equation}
where $\mathcal{N}(0,1)$ is the nominal Gaussian distribution, $\mathcal{U}(-20,20)$ is a uniform distribution with boundary $[-20,20]$, and $p$ is a probability that determines $e_k$ generated by which distribution. Then, we compare the average logarithm likelihood functions $\log \mathcal{L}_{CL}$ and $\log \mathcal{L}_{WLS}$. The corresponding results are shown in Fig. \ref{lh_cmp} and Fig. \ref{lh_cmp_p}. From Fig. \ref{lh_cmp}, one can see that the MKCL is better than the WLS no matter what kernel bandwidth is used with $p=0.2$. From Fig. \ref{lh_cmp_p}, one can see that $\log \mathcal{L}_{CL} \ge \log \mathcal{L}_{WLS}$ always holds when the kernel bandwidth $\sigma$ is optimized with $\sigma^{*}= \arg \max \log \mathcal{L}_{CL}$ (we will introduce the optimization algorithm in the following section) and the net profit of the MKCL over the WLS increases incrementally with the growth of $p$.
\begin{figure}[htbp]
	\centering
	\subfigure[$p=0.2$]{\includegraphics[width=1.75in]{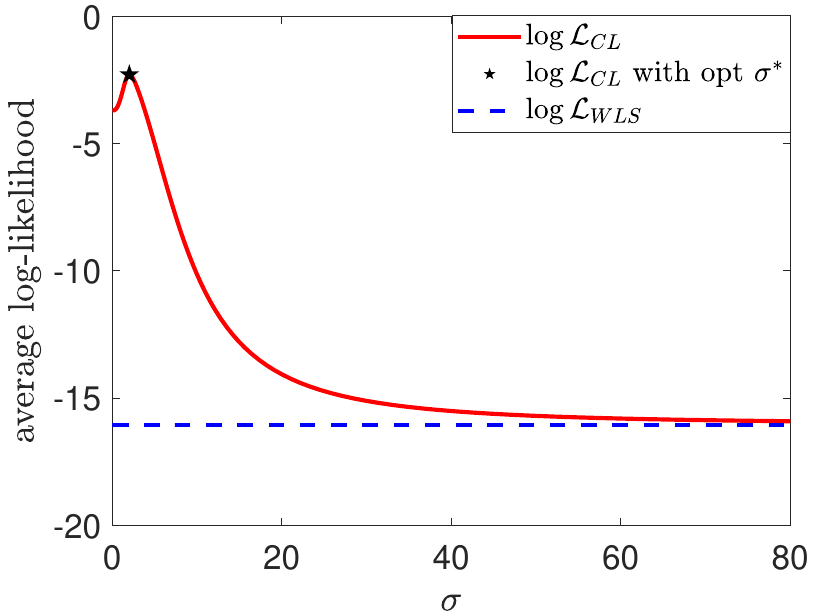}%
		\label{lh_cmp}}
	\subfigure[$0 < p <0.5$]{\includegraphics[width=1.5in]{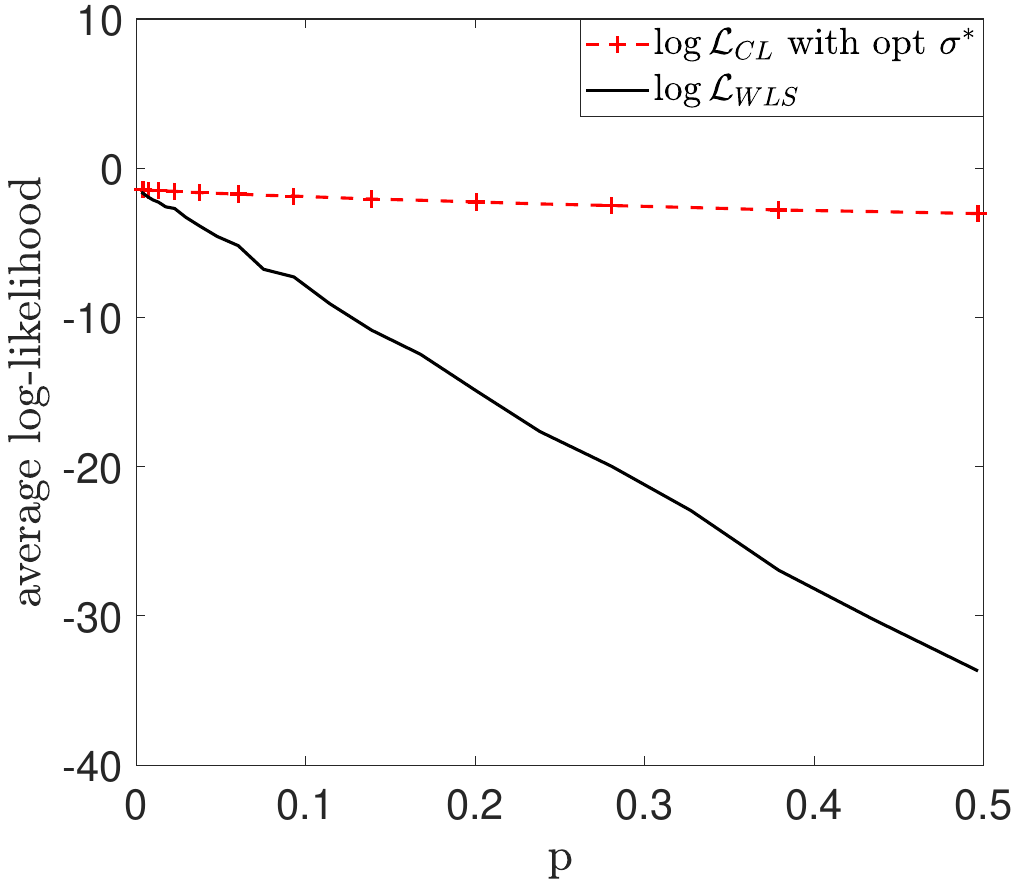}%
		\label{lh_cmp_p}}
	\caption{A comparison of $\log \mathcal{L}_{CL}$ and $\log \mathcal{L}_{WLS}$ when the resuduals follow \eqref{residual}. The logarithm likelihood functions are obtained by $\log\mathcal{L}_{CL}=\frac{1}{N}\sum_{k=1}^{N}p(e_k)$ and $\log\mathcal{L}_{WLS}=\frac{1}{N}\sum_{k=1}^{N}\frac{1}{\sqrt{2\pi}d}\exp\big(-\frac{{e}_{k}^2}{2d^2}\big)$. Moreover, we use $d=1$ in the simulation of Fig. \ref{lh_cmp} and Fig. \ref{lh_cmp_p}. The left figure investigates the values of $\log \mathcal{L}_{CL}$ and $\log \mathcal{L}_{WLS}$ with different $\sigma$ while the rights figure compares $\log \mathcal{L}_{CL}$ and $\log \mathcal{L}_{WLS}$ with respect to different $p$ using optimized $\sigma^{*}$.}
	\label{lh}
\end{figure}
\subsection{Robustness Analysis}
\label{robustness}
We investigate the robustness of a scalar regression problem under the MKCL, i.e., $y_k, X_k, v_k \in \mathbb{R}$ in \eqref{sys}. Specifically, we focus on deriving an upper error bound when outputs are corrupted by outliers with possibly arbitrarily large amplitudes. In the scalar case, the MKCL becomes
\begin{equation}
J_{CL}(\theta) = \sigma^2 \Big(1- \frac{1}{N}\sum_{k=1}^{N}G_{\sigma}\big(\tilde{e}_k\big)\Big)
\label{cl1d}
\end{equation}
with $\tilde{e}_k = \frac{y_k-X_k\theta}{d}= \frac{X_k(\theta^{o}-\theta)+v_k}{d} \in \mathbb{R}$ where $v_k$ can be either a dense noise with a small amplitude or an outlier with a large amplitude. To proceed, we denote $\mathrm{v}_k=\frac{v_k}{d}$ as the normalized noise, $\varepsilon \ge 0$ as a non-negative number,  $I_{N}=\{1,2,\cdots,N\}$ as the sample index set, and $I(\varepsilon)$ as a subset of $I_{N}$ satisfying $\forall k, \mathrm{v}_k  \le \varepsilon$. Moreover, we present the following assumptions:
\begin{assumption}
	$N>|I(\varepsilon)|=M>N/2$ where $|I(\varepsilon)|$ denotes the cardinality of the set $I(\varepsilon)$.
	\label{assump1}
\end{assumption}
\begin{assumption}
	$\exists \zeta>0, s.t., |X_k|>\zeta.$
\end{assumption}
\begin{remark}
	Assumption \ref{assump1} indicates that more than half of the normalized noises $\mathrm{v}_k$ are bounded by $\varepsilon$ while the others can be extremely large (i.e., $\frac{|v_k|}{d} >> \varepsilon$).
\end{remark}
    Then, we have the following theorem.
\begin{theorem}
	If $\sigma> \varepsilon/\sqrt{2\log[{M/(N-M)}]}$, then the optimal solution $\theta$ under $\arg \min J_{CL}(\theta)$ in \eqref{cl1d} satisfies $ |\theta-\theta^{o}| \le \xi $ where 
		\begin{equation}
		\small
		\xi=\frac{d}{\zeta}\left(\sqrt{-2 \sigma^2 \log \left(\exp \left(-\frac{\varepsilon^2}{2 \sigma^2}\right)-\frac{N-M}{M}\right)}+\varepsilon\right).
		\label{bound}
		\end{equation}
	\label{Theorem3}
\end{theorem}
\begin{proof}
	To prove $|\theta-\theta^{o}| \le \xi$ under $\arg \min J_{CL}(\theta)$, it suffices to prove $J_{CL}(\theta)>J_{CL}(\theta^{o})$ for any $\theta$ satisfying $|\theta-\theta^{o}|> \xi$. As $\sigma> \varepsilon/\sqrt{2\log[{M/(N-M)}]}$, it follows that $- \frac{\varepsilon^2}{2\sigma^2} > -\log [M /(N-M)]$ and 
	$\exp(-\frac{\varepsilon^2}{2\sigma^2})> \exp(-\log [M /(N-M)])=\frac{N-M}{M}$. Then, we obtain 
	\begin{equation}
	0<\exp(-\frac{\varepsilon^2}{2\sigma^2})-\frac{N-M}{M}<1
	\end{equation}
	since $\exp(-\frac{\varepsilon^2}{2\sigma^2})\le 1$ and $0<\frac{N-M}{M}<1$. Further, if $|\theta-\theta^{o}|> \xi$, we have $\forall k \in I (\varepsilon)$
	\begin{equation}
	\begin{aligned}
	|\tilde{e}_k| &= \frac{1}{d}|X_k(\theta-\theta^{o})+v_k|\\
	&\overset{3)}{\ge}  \frac{1}{d}\left(|X_k||{\theta}-\theta^{o}|-|v_k|\right)\\
	&\overset{4)}{>}\frac{1}{d}\zeta \xi-\frac{1}{d}|v_k|\\
	&\overset{5)}{>}\frac{1}{d}\zeta \xi-\varepsilon\\
	&=\sqrt{-2 \sigma^2 \log \left(\exp \left(-\frac{\varepsilon^2}{2 \sigma^2}\right)-\frac{N-M}{M}\right)}
	\end{aligned}
	\end{equation}
	where 3) comes from $|X_k(\theta-\theta^{o})+v_k| \ge |X_k(\theta-\theta^{o})|-|v_k|=|X_k||\theta-\theta^{o}|-|v_k|$, 4) comes from $|X_k|>\zeta$ and $|\theta-\theta^{o}|> \xi$, and 5) comes from $\frac{1}{d} |v_k| \le \varepsilon$. Thus, 
	\begin{equation}
	\small
	\begin{aligned}
	\exp\left(-\frac{\tilde{e}_k^2}{2\sigma^2}\right) &< \exp\left(-\frac{-2 \sigma^2 \log \left(\exp \left(-\frac{\varepsilon^2}{2 \sigma^2}\right)-\frac{N-M}{M}\right)}{2\sigma^2}\right)\\
	&=\exp \left(-\frac{\varepsilon^2}{2 \sigma^2}\right)-\frac{N-M}{M},~\forall k \in I(\varepsilon).
	\end{aligned}
	\end{equation}
	Then, we arrive at
	\begin{equation}
	\footnotesize
	\begin{aligned}
	J_{GL}(\theta)&=\sigma^2-\frac{\sigma^2}{N}\Bigg [\sum \limits_{k \in I(\varepsilon)} \exp \left(-\frac{\tilde{e}_k^2}{2\sigma^2}\right) + \sum \limits_{k \notin I(\varepsilon)} \exp\left(-\frac{\tilde{e}_k^2}{2\sigma^2}\right) \Bigg] \\
	& >\sigma^2-\frac{\sigma^2}{N}\Bigg [ \sum \limits_{k \in I(\varepsilon)} \Bigg(\exp \left(-\frac{\varepsilon^2}{2 \sigma^2}\right)-\frac{N-M}{M} \Bigg)\\&
	+  \sum \limits_{k \notin \small{I(\varepsilon)}} \exp\left(-\frac{\tilde{e}_k^2}{2\sigma_2}\right) \Bigg ]\\
	& > \sigma^2 -\frac{\sigma^2}{N}\Bigg [ \sum \limits_{k \in I(\varepsilon)} \left ( \exp \left(-\frac{\varepsilon^2}{2 \sigma^2}\right)-\frac{N-M}{M} \right )+  N-M\Bigg ]\\
	& = \sigma^2 -\frac{\sigma^2}{N}\Bigg [\sum \limits_{k \in I(\varepsilon)} \exp \left(-\frac{\varepsilon^2}{2 \sigma^2}\right) \Bigg ]\\
	& > \sigma^2 -\frac{\sigma^2}{N}\Bigg [\sum \limits_{k \in I(\varepsilon)} \exp \left(-\frac{v_k^2}{2 d^2\sigma^2}\right) \Bigg ]\\
	&> \sigma^2 -\frac{\sigma^2}{N}\Bigg [\sum \limits_{i=1}^{N} \exp \left(-\frac{ v_k^2}{2 d^2\sigma^2}\right) \Bigg ]
	= J_{GL}(\theta^{o}).
	\end{aligned}
	\end{equation}
	This completes the proof.
\end{proof}
\begin{prop}
	The bound $\xi$ first decreases and then increases with the growth of $\sigma$  under $\sigma \in (\varepsilon/\sqrt{2\log[{M/(N-M)}]},\infty)$ and $\varepsilon>0$.
	\label{prop1}
\end{prop}
\begin{proof}
	For $\sigma \in (\varepsilon/\sqrt{2\log[{M/(N-M)}]},\infty)$, taking partial derivative of $\xi$ with respect to $\sigma$ gives
	\begin{equation}
	\small
	\begin{aligned}
	&\frac{\partial \xi}{\partial \sigma}=\frac{d}{\zeta} \Bigg [ \sqrt{-2\log \Big(\exp(-\frac{\varepsilon^2}{2\sigma^2})-\frac{N-M}{M}\Big)} -\\
	&\frac{\varepsilon^2/\sigma^2 \exp(-\frac{\varepsilon^2}{2\sigma^2})}{ \Big(\exp(-\frac{\varepsilon^2}{2\sigma^2})-\frac{N-M}{M}\Big) \sqrt{-2\log \Big(\exp(-\frac{\varepsilon^2}{2\sigma^2})-\frac{N-M}{M}\Big)}} \Bigg].
	\end{aligned}
	\end{equation}
	Denote $p = \exp(-\frac{\varepsilon^2}{2\sigma^2})$ and $q=(N-M)/M$ where $0<q<p<1$, one has
	\begin{equation}
	\small
	\begin{aligned}
	\frac{\partial \xi}{\partial \sigma}&=\frac{d}{\zeta}\bigg(\sqrt{-2\log (p-q)} + \frac{2 p \log(p)}{(p-q)\sqrt{-2\log (p-q)}}\bigg)\\
	&=\frac{d\sqrt{-2\log (p-q)}}{\zeta}\bigg(1- \frac{p \log(p)}{(p-q)\log(p-q)}\bigg).
	\label{dxi}
	\end{aligned}
	\end{equation}
	It is obvious that $\phi(p)=\frac{p \log(p)}{(p-q)\log(p-q)}$ is a monotonically decreasing function of $p$ with $0<q<p<1$. Moreover, $\phi(0_{+}) \to \infty$ and $\phi(1_{-}) \to 0$ and $p$ is a monotonically increasing function of $\sigma$ with $\sigma>0$. Note that $\frac{d\sqrt{-2\log (p-q)}}{\zeta}>0$ always hold. This implies that $\frac{\partial \xi}{\partial \sigma}$ starts from a negative value to a positive value with the growth of $\sigma$. Therefore, $\xi$ first decreases and then increases with the growth of $\sigma$.
\end{proof}
\begin{corollary}
	The ``optimal" kernel bandwidth is obtained when $\phi(p)=1$ in the sense of the lowest upper error bound $\xi$.
\end{corollary}
\begin{prop}
	The absolute upper error $\xi$ is unbounded under the WLS criterion [this corresponds to $\sigma \to \infty$ in \eqref{bound} based on Theorem \ref{theorem1}] for a scalar regression problem.
\end{prop}
\begin{proof}
	According to Theorem \ref{theorem1}, $J_{CL}$ becomes $J_{WLS}$ as $\sigma \to \infty$. Meanwhile, we have $\lim\limits_{\sigma \to \infty} p (\sigma)=1$ and $\phi(1) = 0$ according to the definition of $p(\cdot)$ and $\phi(\cdot)$ in \eqref{dxi}. This implies that $\lim\limits_{\sigma \to \infty}  \frac{\partial \xi}{\partial \sigma} = \frac{d\sqrt{-2\log (1-q)}}{\xi}$ which is a constant. Therefore, the bound $\xi$ would increase linearly with $\sigma$ when $\sigma$ is large and hence unbounded. This completes the proof. 
\end{proof}
\begin{remark}
	The kernel bandwidth $\sigma$ is a critical parameter in problems under the MKC. It should be neither too small nor not too big to minimize the upper error bound. 
\end{remark}
\subsection{Correntropy Parameters Optimization}
\label{kpop}
\begin{algorithm}[!h]
	\caption{Correntropy parameter optimization}
	\label{kpo}
	\KwIn{Parameter vector $\theta$, $\{X_k\}_{k=1}^{N}$ and $\{y_k\}_{k=1}^{N}$}
	\KwOut{Correntropy parameters $\{\sigma\}_{i=1}^{m}$ and $\{d_i\}_{i=1}^{m}$}
	\textbf{Initialization}: Obtain $\{e_k\}_{k=1}^{N} =\{ y_k - f(x_k, \bm{\theta})\}_{k=1}^{N}$, initialize approximated Hessian matrix $H_0$, thresholds $\nu_1$, $\nu_2$, and maximum iteration number $t_{iter}$\\
	\tcc{Optimize correntropy parameters for channel i}
	\For{$i \leftarrow 1$ \KwTo $m$ }
	{ \tcc{Solve $\mathbf{z} = \arg\min f(\mathbf{z})$ as shown in \eqref{obji} with $\mathbf{z} \triangleq (\sigma_i,d_{i})$ using BFGS}  
		\While{$t \le t_{iter}$}
		{Obtain the direction vector $\mathbf {p} _{t}=-H_{t}\nabla f(\mathbf {z} _{t})$\\
			Obtain step length $\alpha_k$ by line search $
			\arg\min  \alpha _{t}=\arg \min f(\mathbf {z} _{t}+\alpha \mathbf {p} _{t})$\\
			$ \mathbf {z} _{t+1}=\mathbf {z} _{t}+\mathbf {s} _{t}$ with $\mathbf {s} _{t}=\alpha _{k}\mathbf {p} _{t}$\\
			$\mathbf {y} _{t}={\nabla f(\mathbf {z} _{t+1})-\nabla f(\mathbf {z} _{t})}$ \\
			$H_{t+1}=H_{t}+{\frac {(\mathbf {s} _{t}^{\mathrm {T} }\mathbf {y} _{t}+\mathbf {y} _{t}^{\mathrm {T} }H_{t}\mathbf {y} _{t})(\mathbf {s} _{t}\mathbf {s} _{t}^{\mathrm {T} })}{(\mathbf {s} _{t}^{\mathrm {T} }\mathbf {y} _{t})^{2}}}-{\frac {H_{t}\mathbf {y} _{t}\mathbf {s} _{t}^{\mathrm {T} }+\mathbf {s} _{t}\mathbf {y} _{t}^{\mathrm {T} }H_{t}}{\mathbf {s} _{t}^{\mathrm {T} }\mathbf {y} _{t}}}$\\
			$t \leftarrow t+1$\\
			\tcc{stop criterion}
			\If{$\|\nabla f(\mathbf {z} _{t+1})\|_2 \le \nu_1$ or $\|\mathbf{s}_{t}\|_2 \le \nu_2$}{
				$Z_i=\mathbf {z} _{t+1}$\\
				\Return}
		}
	}
\end{algorithm}
\textcolor{black}{After building the relationship between MKCL and its induced distribution in Section \ref{mkcpdf}, a remaining question is to optimize the correntropy parameters $\{\sigma_i, d_i\}_{i=1}^{m}$ so that MKCL-induced pdf matches with the practical one.} Based on MLE, one can construct the following problem:
\begin{equation}
\begin{aligned}
&\arg \max \mathcal{L}\Big(\{\sigma_i\}_{i=1}^{m},\{d_i\}_{i=1}^{m},\theta;\{y(k)\}_{k=1}^{N},\{X_k\}_{k=1}^{N}\Big)\\
&=\arg \max \prod_{k=1}^{N}\prod_{i=1}^{m}p\big(e_k(i)\big)\\
&\overset{6)}{=}\arg \min -\sum_{k=1}^{N}\sum_{i=1}^{m}\log \big[p\big(e_k(i)\big)\big]\\
&=\arg \min - \sum_{k=1}^{N}\sum_{i=1}^{m}\Big[\log(c_i-d_i) - \sigma_i^2\big(1-\exp(-\frac{e_k^2(i)}{2d_{i}^2\sigma_i^2})\big)\Big]
\label{mle}
\end{aligned}
\end{equation}
where $6)$ is obtained by taking the negative logarithm function with respect to $\mathcal{L}$. Unfortunately, problem \eqref{mle} cannot be optimized directly since the correntropy parameters $\{\sigma_i\}_{i=1}^{N}$ and $\{d_i\}_{i=1}^{N}$ are coupled with the parameter vector $\theta$. To cope with this issue, an expectation-maximization-like (EM-like) algorithm to used to solve \eqref{mle} alternatingly:
\begin{itemize}
	\item E-step: estimate correntropy parameters $Z_t=(\{\sigma\}_{i=1}^{m},\{d_i\}_{i=1}^{m})$ under current parameter vector $\theta_t$ by solving $Z_t = \arg \max \mathcal{L}\Big(Z_t;\theta_t,\{y(k)\}_{k=1}^{N},\{X_k\}_{k=1}^{N}\Big)$.
	\item M-step: update $\theta_{t+1}$ by solving $\theta_{t+1}=\arg \max \mathcal{L}({\theta}_{t+1};Z_t,\{y(k)\}_{k=1}^{N},\{X_k\}_{k=1}^{N})$. As proved in \eqref{globj}, this procedure is equivalent to minimizing the MKCL as shown in \eqref{cl}.
\end{itemize} 

By assuming that the correntropy parameters at different channels are independent and defining $Z_i\triangleq (\sigma_i,d_i)$, we can execute E step one channel by one channel, i.e.,
\begin{equation}
	\begin{aligned}
	&Z_i = \arg \min f (Z_i) \\
	&=- \arg \min \sum_{k=1}^{N}\Big[\log(c_i-d_i)- \sigma_i^2\big(1-\exp\big(-\frac{e_k^2(i)}{2d_{i}^2\sigma_i^2}\big)\big)\Big]
	\label{kernelopt}
	\end{aligned}
\end{equation} 
where $i=1,2,\ldots,m$. It is worth mentioning that $c_i$ is an implicit function on $Z_i$ and hence cannot be ignored in optimization. In the practical implementation, a truncated distribution of $p(e_k(i))$ is used with sufficiently big domain $e_k(i) \in [-a, a]$ where $a$ can be manually selected so that all $e_k(i)$ is covered by this domain. Then, $c_i$ can be numerically calculated as $c_i=1/\int_{-a}^{a}\frac{1}{d_i}\exp\big(-\sigma_i^2(1-\exp\big(-\frac{e_k(i)^2}{2d_i^2\sigma^2})\big)\big)$ where the numerical integral can be utilized with \emph{integral} command in MATLAB.

The problem \eqref{kernelopt} is a nonlinear objective function and has flat regions when $\sigma_i$ and $d_i$ are big. Hence, gradient-based optimizer is not efficient. To solve this problem, the BFGS algorithm~\cite{b9} is utilized. For channel $i$, we denote $\mathbf{z} \triangleq Z_i$. Then, the minimization problem \eqref{kernelopt} becomes
\begin{equation}
	\begin{aligned}
	&\mathbf{z} = \arg\min f(\mathbf{z})\\
	&=\arg \min - \sum_{k=1}^{N}\sum_{i=1}^{m}\Big[\log(c_i-d_i) - \sigma_i^2\big(1-\exp(-\frac{e_k^2(i)}{2d_{i}^2\sigma_i^2})\big)\Big]
	\label{obji}
	\end{aligned}
\end{equation}
where $f(\mathbf{z}): \mathbb{R}^{2} \rightarrow \mathbb{R}$, and $\mathbf{z} \in \mathbb{R}^{2}$. The detailed algorithm is summarized in Algorithm \ref{kpo} (one can refer to \cite{b9} for more details on nonlinear optimization with BFGS method).

The whole linear regression algorithm with correntropy parameters optimization under the MKC is summarized in Algorithm \ref{lrmkc}, which is called MKC-EM. 

\begin{remark}
	It is worth mentioning that we can solely optimize the kernel bandwidth $\sigma_i$ if the nominal standard deviation $d_i$ is the \emph{a priori} knowledge for channel $i$ in practical applications.  
\end{remark}

\begin{algorithm}[t]
	\caption{Linear regression with MKC-EM}
	\label{lrmkc}
	\KwIn{$\{X_k\}_{k=1}^{N}$ and $\{y_k\}_{k=1}^{N}$}
	\KwOut{Parameter vector $\bm\theta$}
	\textbf{Initialization}: initialize $\theta_0$ by executing Algorithm \ref{fps} and set maximum EM iteration number $t_{iter}$\\
	\While{$t \le t_{iter}$}
	{E step: solve $Z_t = \arg \max \mathcal{L}\big(Z_t;\theta_t,\{X_k\}_{k=1}^{N},\{y_k\}_{k=1}^{N}\big)$ by Algorithm \ref{kpo}\\
		M step: update $\bm\theta_t$ with current $Z_t$ by Algorithm \ref{fps}\\
		$t \leftarrow t+1$}
\end{algorithm}
\begin{remark}
	If we use $\sigma \to \infty$ and ignore the correntropy parameter optimization procedure, MKC-EM degenerates to the conventional WLS regression. If we fix $\sigma \to \infty$ and only optimize $\{d_i\}_{i=1}^{m}$ in Algorithm \ref{lrmkc} (i.e, $Z_t=\{d_i\}_{i=1}^{m}$), then MKC-EM degrades to WLS regression with adaptive weighting matrix.
\end{remark}
\subsection{Convergence Issues}
\label{convergence}
\color{black}
The convergence of the MKC-EM is related to the behavior of the fixed-point solution in Algorithm \ref{fps}, the correntropy parameter optimization in Algorithm \ref{kpo}, and the EM iteration itself. For the fixed-point solution, Chen et al. proved that the convergence of the fixed-point solution is guaranteed under the conventional correntropy if the kernel bandwidth is bigger than a certain value and an initial condition holds in Theorem 2 of \cite{c20}. This theorem can be extended to our Algorithm \ref{fps}, i.e., problem \eqref{fixed} would surely converge to a unique solution if all kernel bandwidths in the MKCL are bigger than a certain threshold. The BFGS algorithm utilized in  Algorithm \ref{kpo} is a popular quasi-Newton method for nonlinear optimization. Its convergence rate is superlinear under some conditions (details are in Theorems 6.5 and 6.6 of \cite{b15}). The convergence behavior of the EM algorithm was discussed in detail in \cite{b16}. Although the EM algorithm possibly converges to local minima or saddle points in some unusual cases \cite{b16}, its performance is satisfactory and converges to the steady state after 2-3 iterations in our algorithm which will be illustrated in illustrative examples in the following section. 
\color{black}
\section{Illustrative Examples}
In this section, we use three examples to demonstrate the effectiveness of the proposed method.
\subsection{Example 1}
Consider the problem 
\begin{equation}
y_k = x_k \theta^{o} + v_k
\end{equation}
where $k$ is the sample index, $x_k = 8 \sin (0.04\pi k)$ is the input, $y_k$ is the output, and $v_k$ is the noise that follows
\begin{equation}
v_k \sim 0.9\mathcal{N}(0, 0.25) + 0.1\mathcal{U}(-20,20).
\end{equation}
\begin{figure}[!h]
	\centerline{\includegraphics[width=6.0cm]{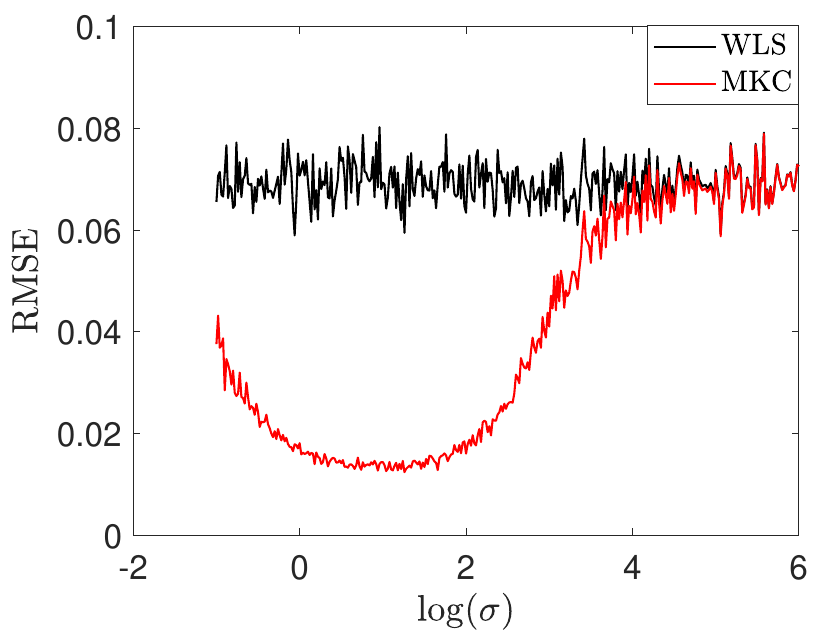}}
	\caption{A comparison of the estimation error under the MKC and WLS with different kernel bandwidths.}
	\label{MKCVSMSE}
\end{figure}

\begin{figure*}[!t]
	\centering
	\subfigure[]{\includegraphics[width=2.25in]{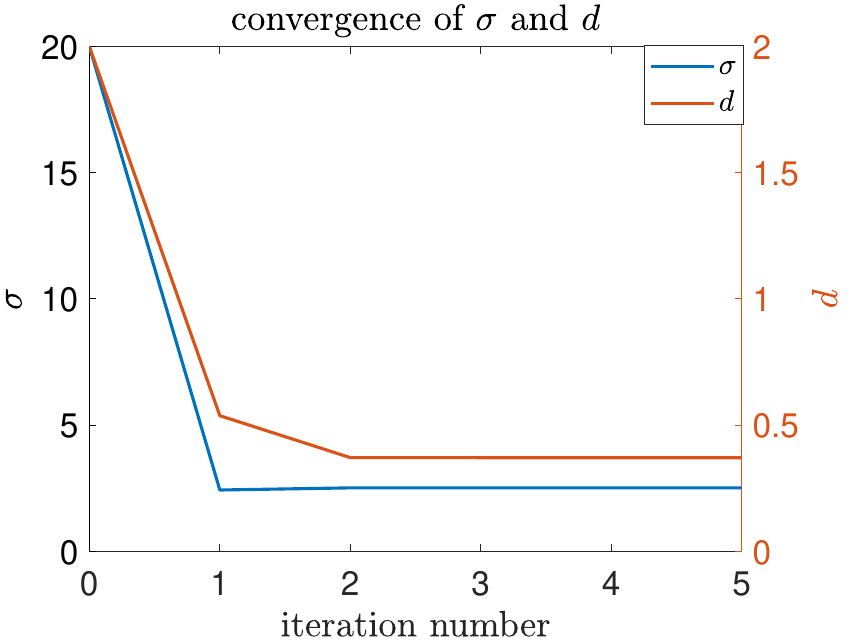}%
		\label{scalar_conv}}
	\subfigure[]{\includegraphics[width=2.25in]{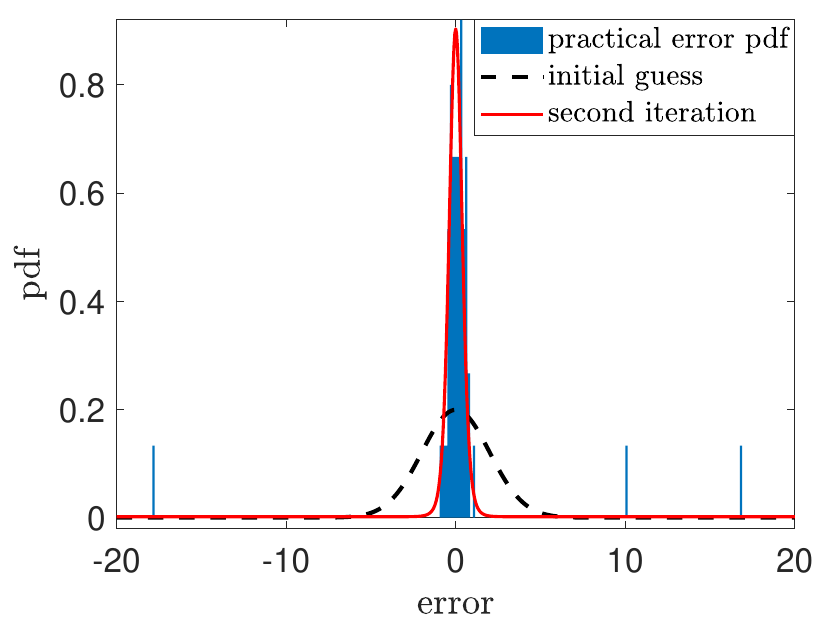}%
		\label{scalar_sigma}}
	\subfigure[]{\includegraphics[width=2.15in]{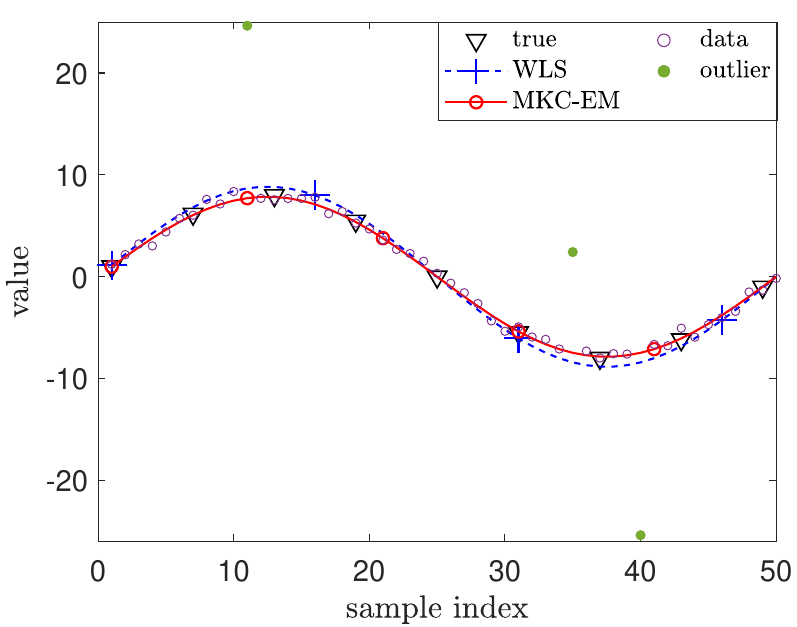}%
		\label{scalar_sin}}
	\caption{The regression results of the MKC-EM and WLS in example 1. The convergence of $\sigma$ and $d$ is shown in Fig. \ref{scalar_conv}. The practical error pdf, induced pdf at 0 and 2 iterations are shown in Fig. \ref{scalar_sigma}. The fitting results under the WLS and MKC-EM are shown in  Fig. \ref{scalar_sin}.}		
	\label{sim1}
\end{figure*}

To investigate the influence of kernel bandwidths in regression under the MKCL, we apply Algorithm \ref{fps} and traverse bandwidths from  $\sigma=\exp(-1.5)$ to $\sigma=\exp(6)$ and execute 200 Monte Carlos runs to obtain the average root mean squared error (RMSE) of $\theta$ [i.e., RMS of $(\theta^{0}-\theta)$]. In the simulation, the nominal standard deviation is set as $d=0.5$. A comparison of the RMSE under the MKCL and WLS is shown in Fig. \ref{MKCVSMSE}. One can see that the RMSE under the MKCL first decreases and then increases with the growth of $\sigma$, and finally coincides with the results of the WLS, which is consistent with Theorem \ref{theorem1} and the log-likelihood investigation in Fig. \ref{lh_cmp}.

We also conduct Algorithm \ref{lrmkc} to obtain both the kernel bandwidth $\sigma$, nominal standard deviation $d$, and the parameter vector ${\theta}$. The initial correntropy parameters are set to be $\sigma=20$ and $d=2$. The convergences of the correntropy parameters, the initial guessed residuals pdf and the MKC-induced pdf at the second iteration, and the comparison of the MKC-EM and WLS are shown in Figs. \ref{scalar_conv},  \ref{scalar_sigma}, and \ref{scalar_sin}, respectively. One can see that MKC-EM converges quickly and its induced pdf approaches the practical one within 2 iterations. In addition, the fitting result of the MKC-EM is very robust to outliers. 
\subsection{Example 2}
\begin{figure*}[htbp]
	\centering
	\subfigure[]{\includegraphics[width=2.5in]{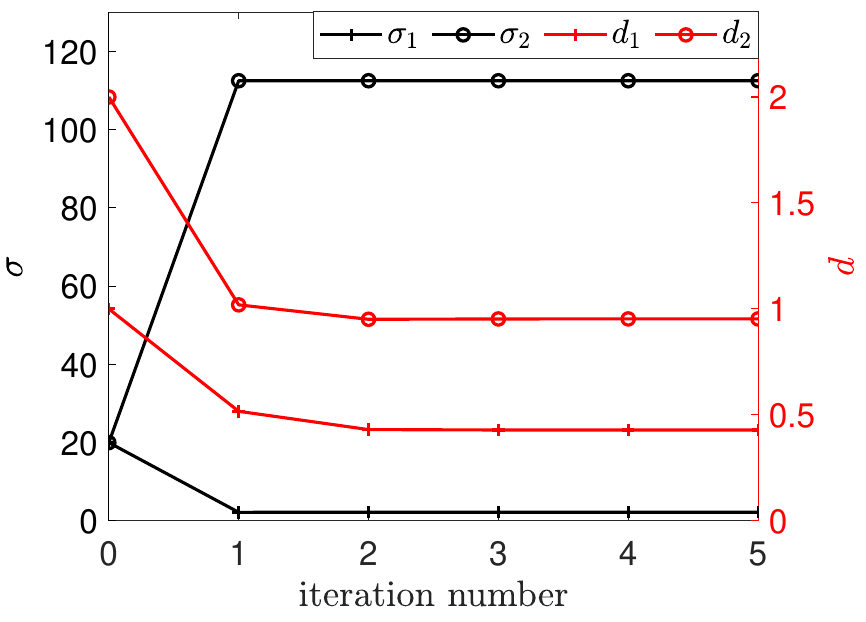}%
		\label{0810_sigma}}
	\subfigure[]{\includegraphics[width=2.43in]{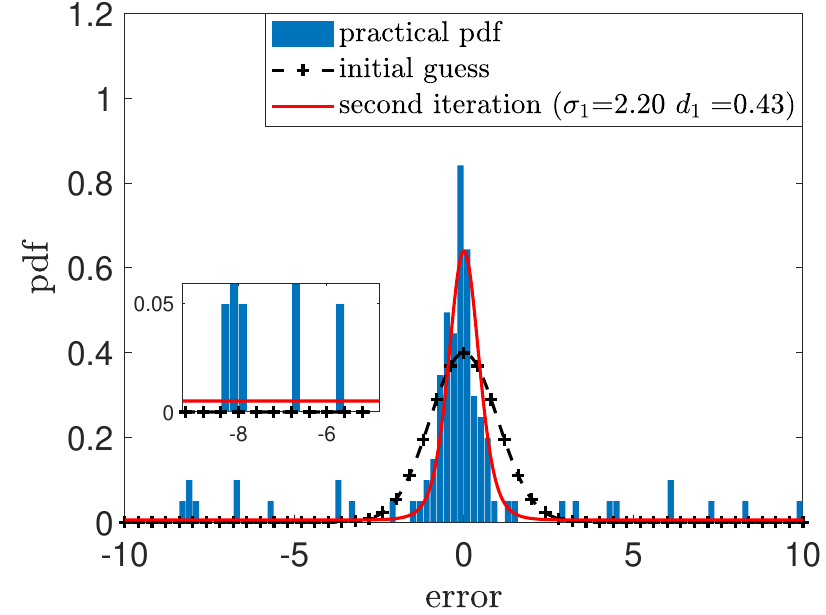}%
		\label{0810_pdf1}}
	
	\subfigure[]{\includegraphics[width=2.5in]{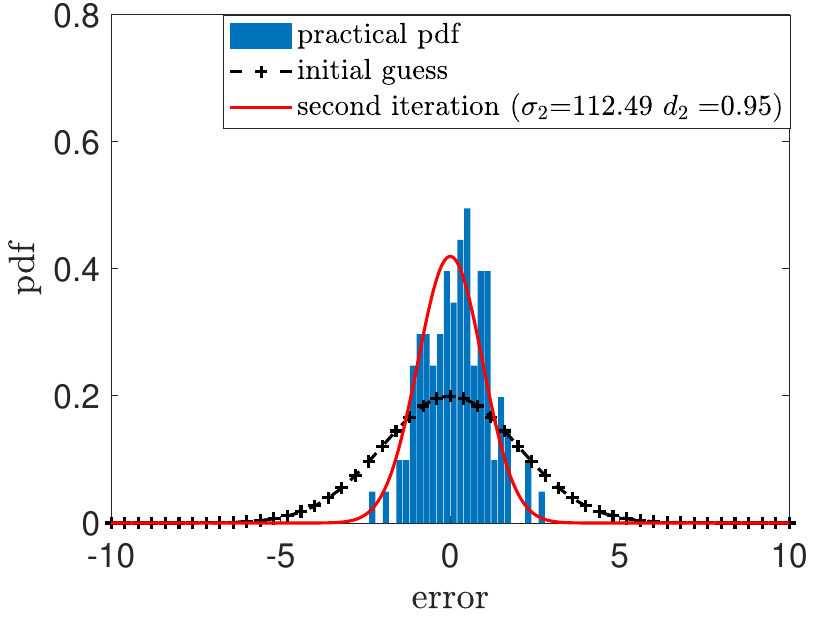}%
		\label{0810_pdf2}}
	\subfigure[]{\includegraphics[width=2.38in]{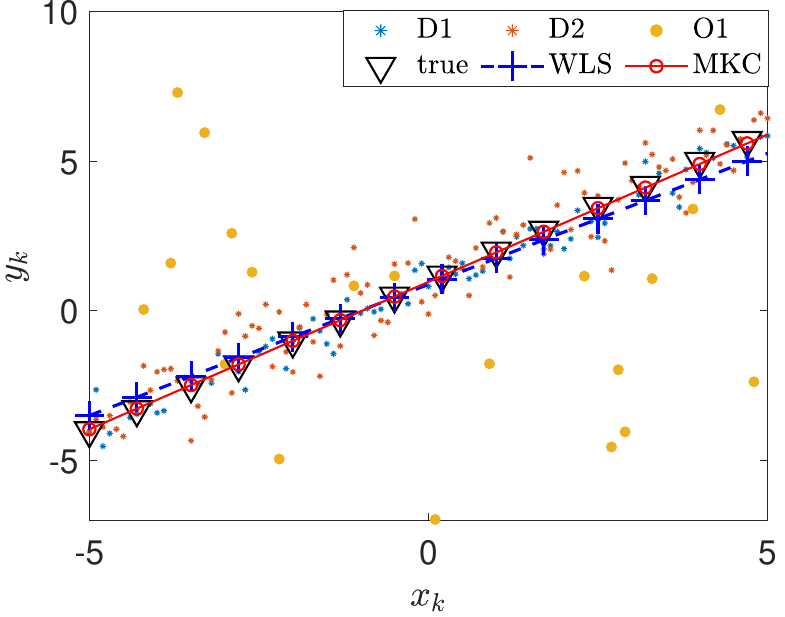}%
		\label{0810_reg}}
	\caption{The regression results with $p_1=0.8$ and $p_2=1.0$. The convergence of $\sigma$ and $d$ are shown in Fig. \ref{0810_sigma}. The practical error pdf, the initial Guess distribution, and the second iteration-induced pdf for channel 1 and 2 in MKC-EM are shown in Fig. \ref{0810_pdf1} and Fig. \ref{0810_pdf2}. The fitting results of the WLS and MKC-EM is shown in Fig. \ref{0810_reg} where D1 and D2 represent the dense noise and O1 represents the outliers generated by channel 1. }		
	\label{sim2}
\end{figure*}

\begin{table*}[!htbp]
	\centering
    \caption{Performance of Different Algorithms in Example 2.}
	\begin{tabular}{cccccc}
		\hline 
		\hline 
		\multirow{2}{*}{case}      & $\left[p_1,p_2\right]$ & WLS $(\|\theta-\theta^{o}\|_2)$            & LAD $(\|\theta-\theta^{o}\|_2)$                   & MKC   $(\|\theta-\theta^{o}\|_2)$                 & MKC-EM    $(\|\theta-\theta^{o}\|_2)$             \\
		&time cost& mean $\pm$ std  &mean $\pm$ std &mean $\pm$ std & mean $\pm$ std \\
		\hline
		\multirow{2}{*}{1} & $[1.0,1.0]$   &  \bm{$0.0395\pm0.0259
		$}& $0.0567\pm0.0336	$ & $0.0508\pm0.0331$ &\bm{$0.0395 \pm 0.0260$}
		  \\
		& time (s)        &  $6.2\times 10^{-5} \pm 6.5\times 10^{-5}$                    &   $0.0035\pm0.0016$                   &         $0.0011\pm0.0003$             &       $0.0264\pm0.0303$               \\
		\multirow{2}{*}{2} &   $[0.8,1.0]$           &        $0.1745\pm0.1011$              &   $0.0585\pm0.0381$                   &         $0.0545\pm0.0356$             &     \bm{$0.0450\pm0.0298$}                 \\
		& time (s)        & $5.9\times 10^{-5} \pm 5.6\times 10^{-5}$                        &     $0.0025\pm0.0008$                 &           $0.0012\pm0.0003$           & $0.0224\pm0.0035$                     \\
		\multirow{2}{*}{3} &   $[1.0,0.8]$              &    $0.1002\pm0.0600$                  &    $0.0567\pm0.0335$                  &   $0.0499\pm0.0309$                   &   \bm{$0.0380\pm0.0260$}                   \\
		& time (s)         &  $6.1\times 10^{-5} \pm 5.7\times 10^{-5}$                    &  $0.0020\pm0.0007$                    &       $0.0012\pm0.0002$               &     $0.0229\pm0.0040$                 \\
		\multirow{2}{*}{4} &     $[0.8,0.8]$            &    $0.1998\pm0.1282$                  &          $0.0606\pm0.0384$            &  $0.0495\pm0.0341$                    &        \bm{$0.0418\pm0.0284$}             \\
		& time (s)        &    $5.8\times 10^{-5} \pm 5.0\times 10^{-5}$                    &  $0.0020\pm0.0007$                    &       $0.0012\pm0.0004$    &       $0.0207\pm0.0033$               \\
		\multirow{2}{*}{5} &    $[0.8,0.6]$             &     $0.3653\pm0.2353$                 &         $0.0711\pm0.0451$             &  $0.0558\pm0.0377$                    &        \bm{$0.0504\pm0.0328$}              \\
		& time (s)        &         $5.8\times 10^{-5} \pm 5.2\times 10^{-5}$          &           $0.0020\pm0.0008$           &        $0.0011\pm0.0002$              &        $0.0215\pm0.0038$              \\
		\multirow{2}{*}{6} &    $[0.6,0.8]$             & $0.4726\pm0.2593$                 &    $0.1314\pm0.0786$           &   $0.0747\pm0.1193$                   &    \bm{$0.0628\pm0.0461$}                  \\
		& time (s)        & $5.8\times 10^{-5} \pm 5.5\times 10^{-5}$                       &  $0.0025\pm0.0008$                    &           $0.0012\pm0.0003$           & $0.0243\pm0.0036$          \\
		\hline
		\hline          
	\end{tabular}
		\label{rmse}
\end{table*}

\begin{figure*}[htbp]
	\centering
	\subfigure[]{\includegraphics[width=2.5in]{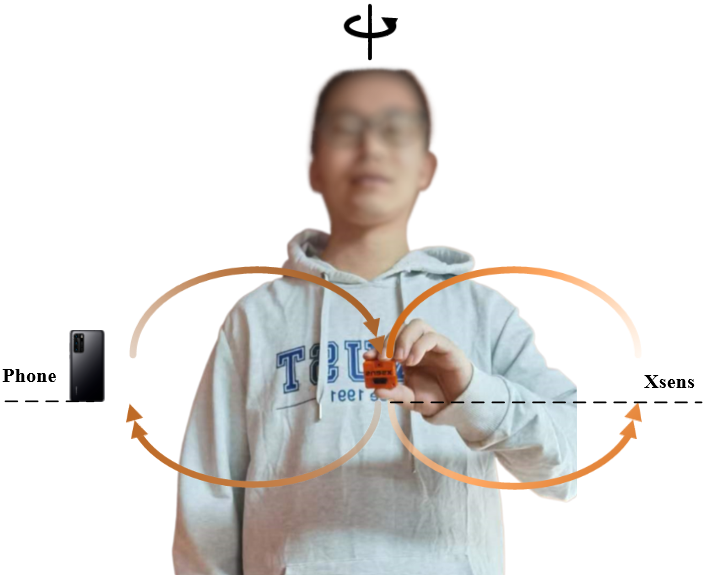}%
		\label{expsetup}}
	\subfigure[]{\includegraphics[width=2.5in]{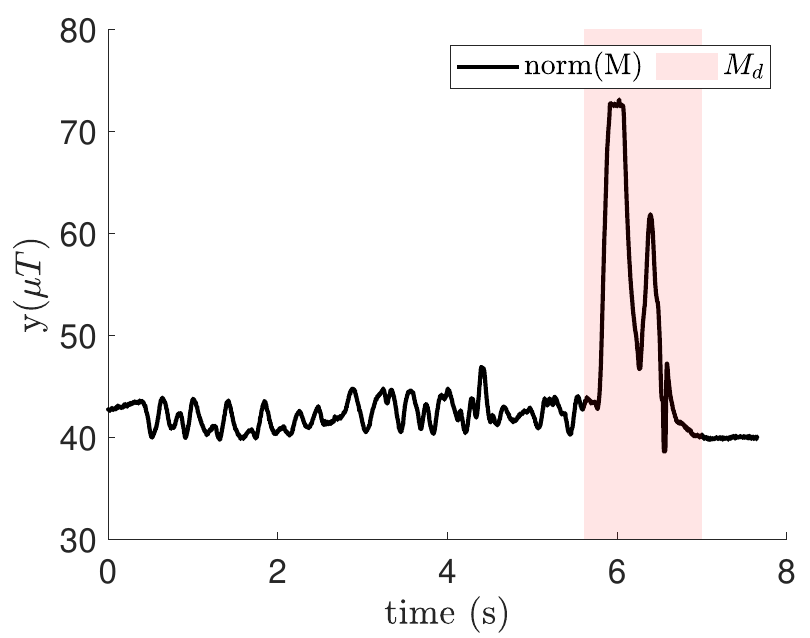}%
		\label{mag_norm}}
	
	\subfigure[]{\includegraphics[width=2.5in]{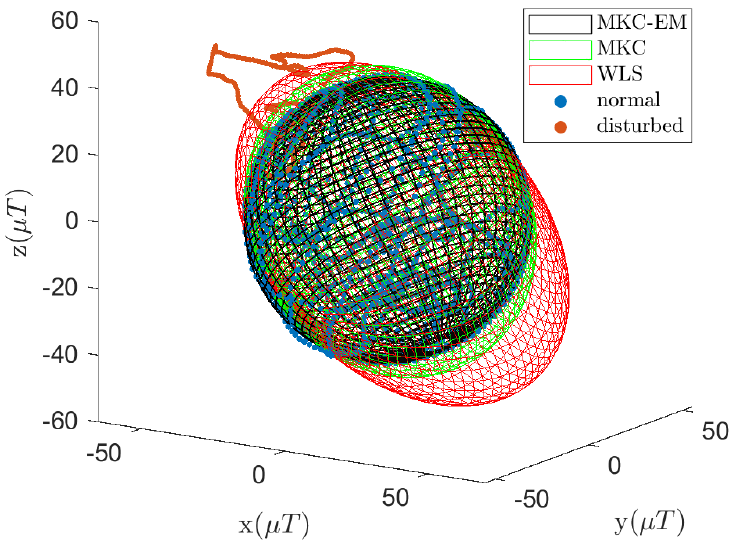}%
		\label{mkc_wls}}
	\subfigure[]{\includegraphics[width=2.5in]{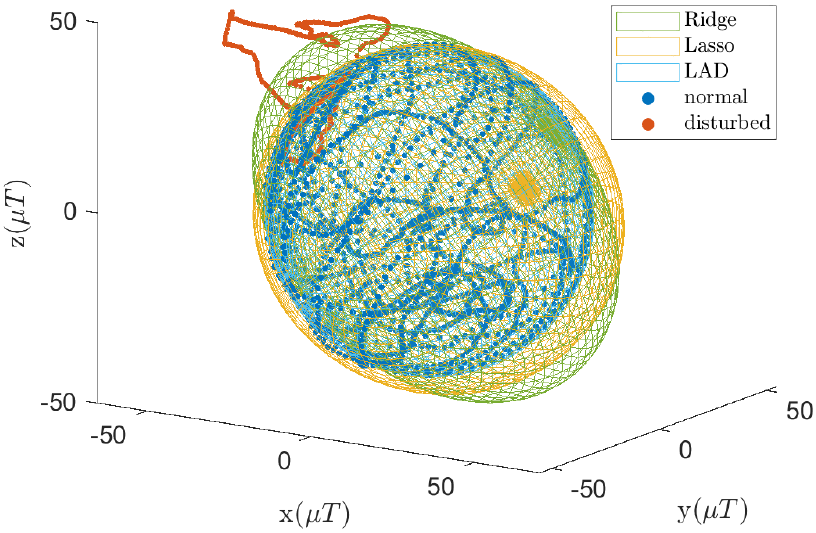}%
		\label{mkc_lasso}}
	\caption{The calibration experiment setup and fitting results. (a) Experimental setup for magnetometer calibration. During the calibration, the Xsens was occasionally close to a mobile phone, resulting in some of the magnetometer measurements being disturbed. (b) The norm of the magnetometer readings. The data that is occasionally disturbed by approaching a mobile phone and the disturbed region is shown in the pink area. (c) Ellipsoid fitting results of MKC-EM, MKC, and WLS. The normal data is shown by the solid blue dot while the disturbed data is shown by the solid brown dot. (d) Ellipsoid fitting results of Ridge regression, and Lasso regression, and LAD regression. The normal data is shown by the solid blue dot while the disturbed data is shown by the solid brown dot.}		
	\label{elliposid}
\end{figure*}
Consider the problem 
\begin{equation}
\begin{bmatrix}
y_{1,k}\\
y_{2,k}
\end{bmatrix} = \begin{bmatrix}1,x_k\\
1,x_k
\end{bmatrix}\begin{bmatrix}
a\\
b
\end{bmatrix}+ \begin{bmatrix}
v_{1,k}\\
v_{2,k}
\end{bmatrix}
\end{equation}
where $X_k=\begin{bmatrix}1,x_k\\
1,x_k
\end{bmatrix}$ is the input at time step $k$, $y_k=\begin{bmatrix}
y_{1,k}\\
y_{2,k}
\end{bmatrix} $ is the corresponding output, and $v_k=\begin{bmatrix}
v_{1,k}\\
v_{2,k}
\end{bmatrix}$ is the noise that follows
\begin{equation}
\begin{aligned}
v_{1,k} &\sim p_1 \mathcal{N}(0, 0.25) + (1-p_1)\mathcal{U}(-10,10)\\
v_{2,k} &\sim p_2 \mathcal{N}(0, 1) + (1-p_2)\mathcal{U}(-20,20).
\label{noise}
\end{aligned}
\end{equation} 
where $p_1$ and $p_2$ are parameters that determine the extent of the heavy tail of the noise. One can see that $y_k$ is heteroscedastic and has different levels of tail for channel $y_{1,k}$ and $y_{2,k}$ (we denote it as channel 1 and channel 2). Our aim is to recover parameter vector $\theta=[a, b]^{T}$ accurately with $\{X_k\}_{k=1}^{N}$ and $\{y_k\}_{k=1}^{N}$ and the ground truth parameter vector is $\theta^{o}=[1,1]^{T}$. We compare the average performance of WLS, MKC, least absolute deviation (LAD)~\cite{b13}, and MKC-EM regression under different $p_1$ and $p_2$ in cases 1 to 6 (as shown in Table \ref{rmse}). To investigate the performance of the proposed algorithm \emph{under different types of heavy-tailed noises}, we use \eqref{noise} in cases 1 to 4, employ $v_{1,k} \sim p_1 \mathcal{N}(0, 0.25) + (1-p_1)\mathcal{N}(0,100)$ and $v_{2,k} \sim p_2 \mathcal{N}(0, 1) + (1-p_2)\mathcal{N}(0,200)$ in case 5, and utilize $v_{1,k} \sim p_1 \mathcal{N}(0, 0.25) + (1-p_1)\mathcal{U}(0,10)$ and $v_{2,k} \sim p_2 \mathcal{N}(0, 1) + (1-p_2)\mathcal{U}(0,20)$ in case 6.

The initial nominal standard deviations are set as $d_1=1$ and $d_2=2$ for the WLS, MKC, and MKC-EM, while the initial kernel bandwidth is set as $\sigma_1=\sigma_2=0.5$ for the MKC and $\sigma_1=\sigma_2=20$ for the MKC-EM. The algorithms are executed in MATLAB on a laptop (Core(TM) i7-1360P, 2.2-GHz CPU, 16-GB RAM) and the sample number is $N=100$. We investigate the Euclidean norm of the parameter vector estimate error $\|\theta-\theta^{o}\|_2$ in 200 Monte Carlo runs, and summarize the corresponding mean $\pm$ standard deviation metrics together with the execution time of different methods in Table \ref{rmse}. One can see that the MKC-EM is identical to WLS (the subtle difference is caused by optimization tolerance) in case 1, and outperforms other others in cases 2 to 6. Moreover, one can observe that WLS degrades significantly with the growth of $p$ but this phenomenon is remarkably alleviated by the MKC-EM. 

We also visualize case 2 in one run in Figs. \ref{0810_sigma}, \ref{0810_pdf1}, \ref{0810_pdf2}, and \ref{0810_reg}. One can see that the kernel bandwidth converges to a small value for channel 1 while it increases to a large value for channel 2 which is in line with the physical interpretation of the kernel bandwidth as shown in Fig. \ref{pdf_ek}. We also observe that the optimized standard deviation is close to the nominal one. Not surprisingly, the proposed MKC-EM has a better pdf matching with the error distribution compared with the initial guess [see Figs. \ref{0810_pdf1} and \ref{0810_pdf2}].
\subsection{Magnetometer Calibration}
We use Xsens MTI-670 which integrates an accelerometer, a gyroscope, and a magnetometer to demonstrate the performance of the proposed algorithm on magnetometer calibration. Specifically, we wave the sensor in a figure-of-eight movement several times accompanied by rotation along the heading direction to ensure that the sensor rotates through all three axes. During the calibration procedure, the magnetometer measurements are occasionally contaminated by external magnetic disturbance (i.e., the mobile phone in Fig. \ref{expsetup}). The experimental setup is shown in Fig. \ref{expsetup}. Our aim is to recover the ellipsoid parameter vector accurately by observing the sampled magnetometer data even if some of the measurements are disturbed. The equations for the ellipsoid are shown in Appendix \ref{ellip}.

The norm of the magnetometer readings is shown in Fig. \ref{mag_norm} with the disturbed area highlighted in pink. The fitting results of the WLS, MKC, and MKC-EM are shown in Fig. \ref{mkc_wls}, and the results of the Ridge regression~\cite{b10}, Lasso regression~\cite{b11}, and LAD regression are shown in Fig. \ref{mkc_lasso}. One can observe that both MKC-EM and LAD fit with the normal data very well and are very robust to magnetic disturbances.

To investigate the performance of different algorithms, we conduct a disturbance-free experiment by removing the surrounding ferromagnetic materials. In this scenario, the measurement noise can be seen as Gaussian and hence the solution of WLS is optimal and is regarded as the ground truth parameter vector $\theta_{ell}^{o}$. Through equations \eqref{ro} and \eqref{S} in Appendix \ref{ellip}, we can subsequently obtain the ground truth ellipsoid center $r_0^{o}$ and semi-axes length vector $s^{o}$. Using this information, we summarize the performance of different algorithms under the disturbed experiment in Table \ref{ellipsoid}. One can see that the result of the MKC-EM is very close to $\theta_{ell}^{o}$ and significantly outperforms the others. 
\begin{table}[htbp]
	\centering
	\caption{Estimation Errors of Different Algorithms.}
	\begin{tabular}{cccc}
		\hline
		\hline
		methods& $\|\theta_{ell}-\theta_{ell}^{o}\|_2$&$\|r_0-r_{0}^{o}\|_2$&$\|s-s^{o}\|_2$  \\
		 LS   & $2.97 \times 10^{-3}$  &$6.329$  &   $17.034$ \\
		 Ridge  & $1.16 \times 10^{-3}$  &$2.096$  & $12.905$\\
		 Lasso  & $2.95 \times 10^{-3}$  &$3.662$  & $8.384$\\
		 LAD  &  $3.18 \times 10^{-4}$ & $0.307$ & $0.411$\\
		 MKC  & $6.45 \times 10^{-4}$  &$0.955$  & $5.170$\\
		 MKC-EM  & \bm{$1.14 \times$ $10^{-4}$}  &\bm{$0.100$}  & \bm{$0.136$}\\
		\hline
	\end{tabular}
\label{ellipsoid}
\end{table}
\section{Conclusion}
This paper investigates the robustness and optimality of the MKC in the context of linear regression. Specifically, we analyze the robustness of the MKC using scalar regression and emphasize the significance of selecting appropriate kernel bandwidths. Additionally, we demonstrate that the MKCL serves as an optimal objective function when the noise distribution conforms to a type of heavy-tailed distribution. To optimize the latent variables (i.e., the correntropy parameters), we develop an EM algorithm that estimates the parameter vector and latent variables alternatingly. Simulations and experiments verify that the proposed algorithm performs very well under Gaussian noise, non-Gaussian noise, and part of channels contaminated by non-Gaussian noise, making it an attractive option when the noise distribution is unknown and probably heavy-tailed. In the future, we will extend this method to the fields of nonlinear regression and classification with heavy-tailed noises.
\section{Appendix}
\subsection{\textcolor{black}{Ellipsoid Fitting}}
\label{ellip}
The implicit equation of a general ellipsoid has
\begin{equation}
	a_{1}x^2+a_{2}y^2+a_{3}z^2+a_{4}xy+a_{5}xz+a_{6}yz+a_{7}x+a_{8}y+a_{9}z=1
	\label{implicit}
\end{equation}
where $(x,y,z)$ is the point defined in a Cartesian coordinate system. By denoting $X_{ell,k}=[x^2,y^2,z^2,xy,xz,yz,x,y,z] \in \mathbb{R}^{1 \times 9}$ and $\theta^{o}_{ell}=[a_{1},a_{2},a_{3},a_{4},a_{5},a_{6},a_{7},a_{8},a_{9}]^{T} \in \mathbb{R}^{9 \times 1}$, $y_{ell,k}=1$, and considering additional noise $v_{ell,k}$, \eqref{implicit} can be written as
\begin{equation}
	y_{ell,k} = X_{ell,k}\theta^{o}_{ell} + v_{ell,k}
	\label{elpf}
\end{equation}
which becomes a linear regression problem. Assume that the parameter vector has been obtained by some optimization algorithms (e.g., the WLS regression in \eqref{WLS} and \eqref{solution}). Then, denote
\begin{equation}\nonumber
	A=\begin{bmatrix}
		a_1, \frac{a_4}{2}, \frac{a_5}{2}\\
		\frac{a_4}{2}, a_2, \frac{a_6}{2}\\
		\frac{a_5}{2}, \frac{a_6}{2},a_3
	\end{bmatrix},
	B=\begin{bmatrix}
		a_7\\
		a_8\\
		a_9
	\end{bmatrix}.
\end{equation}
Based on ~\cite{b8}, the center of the recovered ellipsoid $r_0^{o} \in \mathbb{R}^{3}$ has
\begin{equation}
	r_0^{o}=-\frac{1}{2}A^{-1}B
	\label{ro}
\end{equation}
and the length of the semi-axes vector $s^{o}$ can be obtained by extracting diagonal element the  matrix $S$ with
\begin{equation}
	\begin{aligned}
		S&=\Sigma^{-1/2}\\
		s^{o}&=[s_{11},s_{22},s_{33}]^{T}
	\end{aligned}
	\label{S}
\end{equation}
where $S=[s_{ij}]$, $s^{o} \in \mathbb{R}^{3}$, $\Sigma$ is the diagonal matrix obtained by diagonalizing the matrix $A_1$ with $A_1=Q\Sigma Q^{T}$ and $A_1=\frac{A}{1+r_0^{T}Ar_0}$. To parameterize a general ellipsoid, one has
\begin{equation}
	r=r_0^{o}+Q\Sigma^{-1/2}p
	\label{recov1}
\end{equation}
where $p$ is on a unit sphere and
\begin{equation}
	p=\begin{bmatrix}
		\cos \theta \cos \phi\\
		\cos \theta \sin \phi\\
		\sin \theta
	\end{bmatrix}
	\label{recov2}
\end{equation}
with $-\pi/2<\theta<\pi/2$ and $0<\phi<2\pi$.
\bibliographystyle{IEEEtran}
\bibliography{MultiKerlnelRegression}

\begin{thebibliography}{10}
\providecommand{\url}[1]{#1}
\csname url@samestyle\endcsname
\providecommand{\newblock}{\relax}
\providecommand{\bibinfo}[2]{#2}
\providecommand{\BIBentrySTDinterwordspacing}{\spaceskip=0pt\relax}
\providecommand{\BIBentryALTinterwordstretchfactor}{4}
\providecommand{\BIBentryALTinterwordspacing}{\spaceskip=\fontdimen2\font plus
\BIBentryALTinterwordstretchfactor\fontdimen3\font minus
  \fontdimen4\font\relax}
\providecommand{\BIBforeignlanguage}[2]{{%
\expandafter\ifx\csname l@#1\endcsname\relax
\typeout{** WARNING: IEEEtran.bst: No hyphenation pattern has been}%
\typeout{** loaded for the language `#1'. Using the pattern for}%
\typeout{** the default language instead.}%
\else
\language=\csname l@#1\endcsname
\fi
#2}}
\providecommand{\BIBdecl}{\relax}
\BIBdecl

\bibitem{a1}
Y.-L. Xu and D.-R. Chen, ``Partially-linear least-squares regularized
  regression for system identification,'' \emph{IEEE Transactions on Automatic
  Control}, vol.~54, no.~11, pp. 2637--2641, 2009.

\bibitem{a2}
F.~Wang, M.~R. Gahrooei, Z.~Zhong, T.~Tang, and J.~Shi, ``An augmented
  regression model for tensors with missing values,'' \emph{IEEE Transactions
  on Automation Science and Engineering}, vol.~19, no.~4, pp. 2968--2984, 2022.

\bibitem{a3}
L.~Bako, ``On a class of optimization-based robust estimators,'' \emph{IEEE
  Transactions on Automatic Control}, vol.~62, no.~11, pp. 5990--5997, 2017.

\bibitem{ab4}
C.-C. Peng, J.-J. Huang, and H.-Y. Lee, ``Design of an embedded icosahedron
  mechatronics for robust iterative imu calibration,'' \emph{IEEE/ASME
  Transactions on Mechatronics}, vol.~27, no.~3, pp. 1467--1477, 2021.

\bibitem{a4}
N.~Ozay and M.~Sznaier, ``Hybrid system identification with faulty measurements
  and its application to activity analysis,'' in \emph{Proceedings of the 50th
  IEEE Conference on Decision and Control and European Control Conference},
  2011, pp. 5011--5016.

\bibitem{a5}
H.~Ohlsson and L.~Ljung, ``Identification of switched linear regression models
  using sum-of-norms regularization,'' \emph{Automatica}, vol.~49, no.~4, pp.
  1045--1050, 2013.

\bibitem{ab5}
J.~D. Hol, ``Sensor fusion and calibration of inertial sensors, vision,
  ultra-wideband and gps,'' Ph.D. dissertation, Link{\"o}ping University
  Electronic Press, 2011.

\bibitem{a6}
P.~J. Rousseeuw and A.~M. Leroy, \emph{Robust regression and outlier
  detection}.\hskip 1em plus 0.5em minus 0.4em\relax John Wiley \& Sons, 2005.

\bibitem{a7}
P.~J. Rousseeuw, ``Least median of squares regression,'' \emph{Journal of the
  American statistical association}, vol.~79, no. 388, pp. 871--880, 1984.

\bibitem{a8}
A.~Y. Aravkin, B.~M. Bell, J.~V. Burke, and G.~Pillonetto, ``An $\ell
  _{1}$-laplace robust kalman smoother,'' \emph{IEEE Transactions on Automatic
  Control}, vol.~56, no.~12, pp. 2898--2911, 2011.

\bibitem{a9}
C.~L. Nikias and M.~Shao, \emph{Signal processing with alpha-stable
  distributions and applications}.\hskip 1em plus 0.5em minus 0.4em\relax
  Wiley-Interscience, 1995.

\bibitem{a10}
R.~A. Maronna, R.~D. Martin, V.~J. Yohai, and M.~Salibi{\'a}n-Barrera,
  \emph{Robust statistics: theory and methods (with R)}.\hskip 1em plus 0.5em
  minus 0.4em\relax John Wiley \& Sons, 2019.

\bibitem{a11}
W.~Liu, P.~P. Pokharel, and J.~C. Principe, ``Correntropy: Properties and
  applications in non-gaussian signal processing,'' \emph{IEEE Transactions on
  Signal Processing}, vol.~55, no.~11, pp. 5286--5298, 2007.

\bibitem{a12}
A.~Aravkin, J.~V. Burke, L.~Ljung, A.~Lozano, and G.~Pillonetto, ``Generalized
  kalman smoothing: Modeling and algorithms,'' \emph{Automatica}, vol.~86, pp.
  63--86, 2017.

\bibitem{a13}
B.~Chen, X.~Liu, H.~Zhao, and J.~C. Principe, ``Maximum correntropy kalman
  filter,'' \emph{Automatica}, vol.~76, pp. 70--77, 2017.

\bibitem{a14}
L.~Bako, ``Robustness analysis of a maximum correntropy framework for linear
  regression,'' \emph{Automatica}, vol.~87, pp. 218--225, 2018.

\bibitem{a15}
B.~Chen, X.~Wang, Y.~Li, and J.~C. Principe, ``Maximum correntropy criterion
  with variable center,'' \emph{IEEE Signal Processing Letters}, vol.~26,
  no.~8, pp. 1212--1216, 2019.

\bibitem{a16}
B.~Chen, J.~Liang, N.~Zheng, and J.~C. Príncipe, ``Kernel least mean square
  with adaptive kernel size,'' \emph{Neurocomputing}, vol. 191, pp. 95--106,
  2016.

\bibitem{a17}
M.~V. Kulikova, ``Chandrasekhar-based maximum correntropy kalman filtering with
  the adaptive kernel size selection,'' \emph{IEEE Transactions on Automatic
  Control}, vol.~65, no.~2, pp. 741--748, 2020.

\bibitem{a18}
G.~Wang, Y.~Zhang, and X.~Wang, ``Maximum correntropy rauch–tung–striebel
  smoother for nonlinear and non-gaussian systems,'' \emph{IEEE Transactions on
  Automatic Control}, vol.~66, no.~3, pp. 1270--1277, 2021.

\bibitem{a19}
B.~Chen, L.~Xing, H.~Zhao, N.~Zheng, and J.~C. Prı´ncipe, ``Generalized
  correntropy for robust adaptive filtering,'' \emph{IEEE Transactions on
  Signal Processing}, vol.~64, no.~13, pp. 3376--3387, 2016.

\bibitem{a20}
B.~Chen, Y.~Xie, X.~Wang, Z.~Yuan, P.~Ren, and J.~Qin, ``Multikernel
  correntropy for robust learning,'' \emph{IEEE Transactions on Cybernetics},
  vol.~52, no.~12, pp. 13\,500--13\,511, 2022.

\bibitem{b20}
A.~Singh and J.~C. Principe, ``Using correntropy as a cost function in linear
  adaptive filters,'' in \emph{Proceedings of the International Joint
  Conference on Neural Networks}, 2009, pp. 2950--2955.

\bibitem{c20}
B.~Chen, J.~Wang, H.~Zhao, N.~Zheng, and J.~C. Príncipe, ``Convergence of a
  fixed-point algorithm under maximum correntropy criterion,'' \emph{IEEE
  Signal Processing Letters}, vol.~22, no.~10, pp. 1723--1727, 2015.

\bibitem{d20}
R.~He, B.-G. Hu, W.-S. Zheng, and X.-W. Kong, ``Robust principal component
  analysis based on maximum correntropy criterion,'' \emph{IEEE Transactions on
  Image Processing}, vol.~20, no.~6, pp. 1485--1494, 2011.

\bibitem{e20}
Q.~Zhang and H.~Muhlenbein, ``On the convergence of a class of estimation of
  distribution algorithms,'' \emph{IEEE Transactions on Evolutionary
  Computation}, vol.~8, no.~2, pp. 127--136, 2004.

\bibitem{a21}
B.~Chen, L.~Xing, H.~Zhao, S.~Du, and J.~C. Príncipe, ``Effects of outliers on
  the maximum correntropy estimation: A robustness analysis,'' \emph{IEEE
  Transactions on Systems, Man, and Cybernetics: Systems}, vol.~51, no.~6, pp.
  4007--4012, 2021.

\bibitem{a22}
F.~Huang, J.~Zhang, and S.~Zhang, ``Adaptive filtering under a variable kernel
  width maximum correntropy criterion,'' \emph{IEEE Transactions on Circuits
  and Systems II: Express Briefs}, vol.~64, no.~10, pp. 1247--1251, 2017.

\bibitem{a23}
B.~Chen, X.~Wang, Y.~Li, and J.~C. Principe, ``Maximum correntropy criterion
  with variable center,'' \emph{IEEE Signal Processing Letters}, vol.~26,
  no.~8, pp. 1212--1216, 2019.

\bibitem{a24}
S.~Li, D.~Shi, W.~Zou, and L.~Shi, ``Multi-kernel maximum correntropy kalman
  filter,'' \emph{IEEE Control Systems Letters}, vol.~6, pp. 1490--1495, 2022.

\bibitem{a25}
S.~Li, P.~Duan, D.~Shi, W.~Zou, P.~Duan, and L.~Shi, ``Compact maximum
  correntropy-based error state kalman filter for exoskeleton orientation
  estimation,'' \emph{IEEE Transactions on Control Systems Technology}, pp.
  1--8, 2022.

\bibitem{a26}
S.~Li, L.~Li, D.~Shi, W.~Zou, P.~Duan, and L.~Shi, ``Multi-kernel maximum
  correntropy kalman filter for orientation estimation,'' \emph{IEEE Robotics
  and Automation Letters}, vol.~7, no.~3, pp. 6693--6700, 2022.

\bibitem{b12}
I.~Goodfellow, Y.~Bengio, and A.~Courville, \emph{Deep learning}.\hskip 1em
  plus 0.5em minus 0.4em\relax MIT press, 2016.

\bibitem{b5}
X.~Zhang, C.~Zhou, F.~Chao, C.-M. Lin, L.~Yang, C.~Shang, and Q.~Shen,
  ``Low-cost inertial measurement unit calibration with nonlinear scale
  factors,'' \emph{IEEE Transactions on Industrial Informatics}, vol.~18,
  no.~2, pp. 1028--1038, 2021.

\bibitem{b6}
C.-C. Peng, J.-J. Huang, and H.-Y. Lee, ``Design of an embedded icosahedron
  mechatronics for robust iterative imu calibration,'' \emph{IEEE/ASME
  Transactions on Mechatronics}, vol.~27, no.~3, pp. 1467--1477, 2021.

\bibitem{b2}
C.~Grund, J.~Tanke, and J.~Gall, ``Ellipose: Stereoscopic 3d human pose
  estimation by fitting ellipsoids,'' in \emph{Proceedings of the IEEE/CVF
  Winter Conference on Applications of Computer Vision}, 2023, pp. 2871--2881.

\bibitem{b9}
N.~Jorge and J.~W. Stephen, \emph{Numerical optimization}.\hskip 1em plus 0.5em
  minus 0.4em\relax Spinger, 2006.

\bibitem{b15}
J.~Nocedal and S.~J. Wright, \emph{Numerical optimization}.\hskip 1em plus
  0.5em minus 0.4em\relax Springer, 1999.

\bibitem{b16}
G.~J. McLachlan and T.~Krishnan, \emph{The EM algorithm and extensions}.\hskip
  1em plus 0.5em minus 0.4em\relax John Wiley \& Sons, 2007.

\bibitem{b13}
D.~Pollard, ``Asymptotics for least absolute deviation regression estimators,''
  \emph{Econometric Theory}, vol.~7, no.~2, pp. 186--199, 1991.

\bibitem{b10}
G.~C. McDonald, ``Ridge regression,'' \emph{Wiley Interdisciplinary Reviews:
  Computational Statistics}, vol.~1, no.~1, pp. 93--100, 2009.

\bibitem{b11}
J.~Ranstam and J.~Cook, ``Lasso regression,'' \emph{Journal of British
  Surgery}, vol. 105, no.~10, pp. 1348--1348, 2018.

\bibitem{b8}
B.~Bertoni, ``Multi-dimensional ellipsoidal fitting,'' \emph{Department of
  Physics, South Methodist University, Tech. Rep. SMU-HEP-10-14}, 2010.

\end{thebibliography}
\end{document}